\documentclass[a4paper, 11pt]{article}
\usepackage[margin=.95in,dvips]{geometry}
\usepackage[bottom]{footmisc}
\usepackage[affil-it]{authblk}
\usepackage{verbatim}
\usepackage{float}
\usepackage{amsmath}
\usepackage{tikz}
\usepackage{etoolbox}

\pgfdeclarelayer{back}
\pgfdeclarelayer{front}
\pgfsetlayers{back,main,front}

\usepackage{hyperref}
\hypersetup{
  colorlinks   = true, 
  urlcolor     = blue, 
  linkcolor    = blue, 
  citecolor   = blue 
}
\usepackage{enumitem}
\setlist[enumerate]{itemsep=2.0pt plus 1.0 pt minus 0.5pt, topsep=4.0pt plus 2.0 pt minus 1.0pt}
\setlist[itemize]{itemsep=2.0pt plus 1.0 pt minus 0.5pt, topsep=4.0pt plus 2.0 pt minus 1.0pt}
\usepackage{calc}
\usepackage{graphicx}
\usepackage{subcaption}
\usepackage{color, soul}

\makeatletter
\pgfkeys{%
  /tikz/on layer/.code={
    \pgfonlayer{#1}\begingroup
    \aftergroup\endpgfonlayer
    \aftergroup\endgroup
  },
  /tikz/node on layer/.code={
    \gdef\node@@on@layer{%
      \setbox\tikz@tempbox=\hbox\bgroup\pgfonlayer{#1}\unhbox\tikz@tempbox\endpgfonlayer\egroup}
    \aftergroup\node@on@layer
  },
  /tikz/end node on layer/.code={
    \endpgfonlayer\endgroup\endgroup
  }
}
\def\node@on@layer{\aftergroup\node@@on@layer}
\def\calign@preamble{%
   &\hfil\strut@
    \setboxz@h{\@lign$\m@th\displaystyle{##}$}%
    \ifmeasuring@\savefieldlength@\fi
    \set@field
    \hfil
    \tabskip\alignsep@
}
\let\cmeasure@\measure@
\patchcmd\cmeasure@{\divide\@tempcntb\tw@}{}{}{}
\patchcmd\cmeasure@{\divide\@tempcntb\tw@}{}{}{}
\patchcmd\cmeasure@{\ifodd\maxfields@
  \global\advance\maxfields@\@ne
  \fi}{}{}{}    
\newenvironment{calign}
{%
  \let\align@preamble\calign@preamble
  \let\measure@\cmeasure@
  \align
}
{%
  \endalign
}  
\makeatother

\usepackage{marginnote}
\usepackage{color,soulutf8}

\newcounter{commcounter}
\usepackage{amsmath}
\usepackage{amsthm, slashed}
\usepackage[capitalize]{cleveref}
\usepackage{amssymb}
\usepackage{array}
\usepackage{xfrac}
\usepackage{braket}
\usepackage{bbm}
\usepackage{bm}
\usepackage{revsymb} 
\usepackage{mathtools}
\usepackage{mathrsfs}
\usepackage{accents}

\allowdisplaybreaks

\usepackage[square, numbers, comma, sort&compress]{natbib}

\let\OLDthebibliography\thebibliography
\renewcommand\thebibliography[1]{
  \OLDthebibliography{#1}
  \setlength{\parskip}{0pt}
  \setlength{\itemsep}{0pt plus 0.5ex}
}

\DeclareMathOperator{\Tr}{Tr}

\usepackage{amsthm}

\theoremstyle{plain}
\newtheorem{theorem}{Theorem}[section]
\newtheorem{lemma}[theorem]{Lemma}
\theoremstyle{definition}
\newtheorem{defn}[theorem]{Definition}

\renewcommand\tilde[1]{\widetilde{#1}}
\tikzset{redregion/.style={fill=red, fill opacity=0.3, draw=none}}
\tikzset{blueregion/.style={fill=blue, fill opacity=0.3, draw=none}}
\tikzset{block/.style={draw, fill=blue!40, minimum width=0.55cm, minimum height=0.55cm, rounded corners=2pt, node on layer=front, thick, inner sep=-5pt}}
\tikzset{blockflip/.style={block, fill=black!40}}
\tikzset{measure/.style={block, fill=red}}

\begin{document}

\title{\bf From dual-unitary to biunitary:
\\
a 2-categorical model for exactly-solvable
\\many-body quantum dynamics}

\author[1]{Pieter W. Claeys\footnote{\href{mailto:claeys@pks.mpg.de}{claeys@pks.mpg.de}}}
\author[2]{Austen Lamacraft\footnote{\href{mailto:al200@cam.ac.uk}{al200@cam.ac.uk}}}
\author[3]{Jamie Vicary\footnote{\href{mailto:jamie.vicary@cl.cam.ac.uk}{jamie.vicary@cl.cam.ac.uk}}}

\affil[1]{\small Max Planck Institute for the Physics of Complex Systems, 01187 Dresden, Germany}
\affil[2]{\small TCM Group, Cavendish Laboratory, University of Cambridge, Cambridge CB3 0HE, UK}
\affil[3]{\small Computer Laboratory, University of Cambridge, Cambridge, CB3 0FD, UK}

\date{}

\maketitle

%\vspace{-0.5cm}

\begin{abstract}
Dual-unitary brickwork circuits are an exactly-solvable model for many-body chaotic quantum systems, based on 2-site gates which are unitary in both the time and space directions. Prosen has recently described an alternative  model called \textit{dual-unitary interactions round-a-face}, which we here call \textit{clockwork}, based on 2-controlled 1-site unitaries composed in a non-brickwork structure, yet with many of the same attractive global properties. We present a 2\-categorical framework that simultaneously generalizes these two existing models, and  use it to show that brickwork and clockwork circuits can interact richly, yielding new types of generalized heterogeneous circuits. We show that these interactions are governed by quantum combinatorial data, which we precisely characterize. These generalized circuits remain exactly-solvable and we show that they retain the attractive features of the original models such as single-site correlation functions vanishing everywhere except on the causal light-cone. Our framework allows us to directly extend the notion of solvable initial states to these biunitary circuits, and we show these circuits demonstrate maximal entanglement growth and exact thermalization after finitely many time steps.
\end{abstract}

\vspace{0.5cm}

%--------------------------------------------------------------------------------------------------------------------------------------------------------
\section{Introduction}
\label{sec:intro}

The dynamics of isolated many-body quantum systems remains a  complex problem, and exact solutions are both scarce and typically non-representative of the chaotic dynamics of generic quantum systems.
In recent years unitary circuit models have gained intense attention as a paradigmatic model of unitary many-body dynamics governed by local interactions \cite{nahum_quantum_2017,khemani_operator_2018,von_keyserlingk_operator_2018,nahum_operator_2018,chan_solution_2018,rakovszky_sub-ballistic_2019,friedman_spectral_2019,garratt_local_2021}. Such circuits mimic the dynamics generated by a local Hamiltonian on a one-dimensional lattice, with the `dynamics' of a unitary circuit taking place in discrete time. 

A special class of \textit{dual-unitary} circuits was recently identified, characterized by the property that the bulk dynamics remains unitary when exchanging the roles of space and time \cite{bertini_exact_2019,gopalakrishnan_unitary_2019}. This duality endows the circuits with many remarkable properties, including  analytically tractable correlation dynamics~\cite{bertini_exact_2019,piroli_exact_2020,gutkin_exact_2020,claeys_ergodic_2021,kos_correlations_2021}, out-of-time-order correlators \cite{claeys_maximum_2020,bertini_scrambling_2020}, and maximal entanglement growth~\cite{piroli_exact_2020,reid_entanglement_2021,zhou_maximal_2022,foligno_growth_2022}. While exactly-solvable many-body models typically require integrability and hence non-chaotic dynamics, these dual-unitary models are in a sense `maximally chaotic', allowing some generic features of quantum many-body chaos to be studied exactly~\cite{bertini_exact_2018, bertini_random_2021, kos_chaos_2021, fritzsch_eigenstate_2021, fritzsch_boundary_2021,ho_exact_2022, claeys_emergent_2022, ippoliti_dynamical_2022,fritzsch_boundary_2023}.

\newcommand\vc[1]{\begin{tabular}{c}#1\end{tabular}}

Most recent studies of dual-unitary dynamics have focused on 2-site dual-unitary gates arranged in a regular brickwork pattern. An alternative model has recently been proposed by Prosen~\cite{prosen_many-body_2021} who showed that dual-unitary interactions could also be arranged `round-a-face', a model we here name \emph{clockwork}\footnote{Since the hands of a clock go round-a-face.}, with a circuit representation in terms of 1-site unitaries controlled by both adjacent systems. We illustrate these circuits as follows:
\begin{calign}
\label{dubr_durf}
\begin{aligned}
\includegraphics[width=5cm]{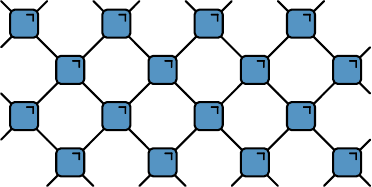}
\end{aligned}
&
\begin{aligned}
\includegraphics[width=5cm]{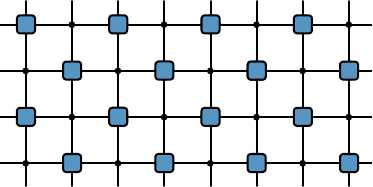}
\end{aligned}
\\ \nonumber
\textrm{\em (a) Conventional brickwork circuit}
&
\textrm{\em (b) Conventional clockwork circuit}
\end{calign}
These models have been shown to have similar global properties, supporting notions of unitarity along both the time (vertical) and space (horizontal) direction, and with vanishing single-site correlation functions everywhere except on the causal light cone. A common generalisation is therefore strongly suggested, and indeed Prosen writes ``it is an interesting open question if one can find mappings between [brickwork and clockwork] circuit models on an abstract level''.

Here we  present such a generalisation based on \textit{biunitary connections}, algebraic structures with a variety of applications in 2\-categorical linear algebra. We reason about these structures using the \textit{shaded calculus}~\cite{reutter_biunitary_2019, Lauda2006, Jones1999}, a powerful graphical system analogous to traditional tensor notation, but with the added feature of shadings assigned to certain diagram regions. In the simplest representation scheme, these shadings represent the assignment of a \textit{finite indexing set} to the region, and any wires and vertices that border the region are then indexed by that set. If any particular region is left unshaded, that means the region is assigned a 1\-element set; in this way the shaded calculus subsumes the traditional tensor notation, since the wires and vertices bordering such a region become trivially indexed. For every shaded calculus diagram we can compute a traditional representation in tensor notation; this is a nontrivial process which increases the apparent circuit complexity,  since we must add controlled gates to encode the indexing behaviour. 

Full details of the shaded calculus must wait until Section~\ref{sec:biunitary}. However, we can already use the shaded calculus to give a clear intuitive sense of our main results. Imposing for now spacetime homogeneity, there are two possible shaded circuit structures, either completely unshaded or completely shaded, as follows:
\begin{calign}
\label{fig:circuits_intro}
\begin{aligned}
\includegraphics[width=5cm]{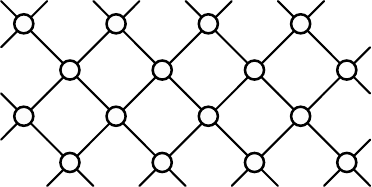}
\end{aligned}
&
\begin{aligned}
\includegraphics[width=5cm]{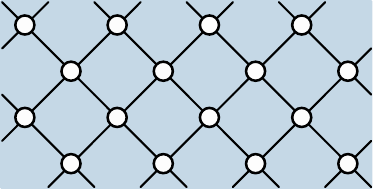}
\end{aligned}
\\ \nonumber
\textrm{\em (a) Shaded calculus brickwork circuit}
&
\textrm{\em (b) Shaded calculus clockwork circuit}
\end{calign}
Both these diagrams are drawn in the shaded calculus, which we indicate by using circles for the vertices, to contrast with the conventional circuit diagrams where we draw vertices as squares. Diagram~(2)(a)\ is a degenerate case of the shaded calculus since every region is trivially shaded; it can therefore be interpreted directly as an ordinary tensor diagram, yielding the traditional brickwork circuit (1)(a). In contrast, diagram (2)(b) has nontrivial shading, and computing its associated tensor representation gives precisely the clockwork circuit (1)(b).

This shaded calculus representation therefore yields a structural unification of the brickwork and clockwork models; while as conventional circuits they have very different structures, in the shaded calculus their representation is uniform, with 4-valent vertices stacked in a brickwork pattern. Furthermore, the dual unitarity property, which requires different definitions in the conventional brickwork and clockwork cases, is unified by the single concept of biunitarity in the shaded calculus representation.

However, our scheme offers more than notational unification, since spacetime homogeneity is not a requirement. Dropping this condition allows us to explore a rich family of exactly-solvable circuits, which can  include more complex interaction patterns and geometries, which have not previously been described. As a first example, we consider a simple diagonal boundary between clockwork and brickwork regions, giving on the left the shaded calculus representation, and on the right the conventional circuit model:

\begin{calign}
\label{fig:diagonal_intro}
\begin{aligned}
\includegraphics[width=7cm]{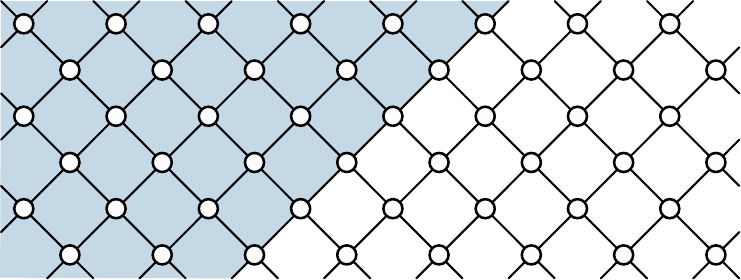}
\end{aligned}
&
\begin{aligned}
\includegraphics[width=7cm]{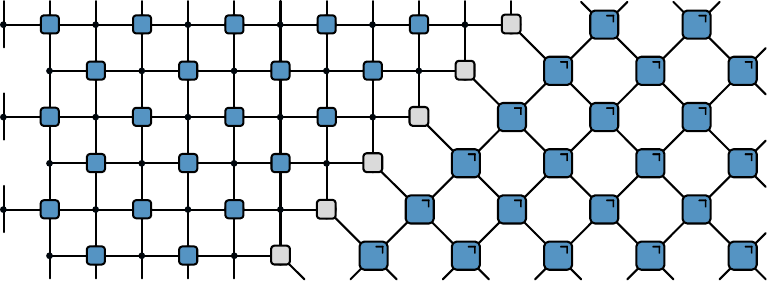}
\end{aligned}
\\ \nonumber
\textrm{\em (a) Shaded calculus representation}
&
\textrm{\em (b) Unitary circuit representation}
\end{calign}
This describes a \textit{dynamical boundary} between a clockwork region on the left of each picture, and a brickwork region on the right, separated by a boundary which moves left-to-right over time. In the diagram on the left, along the diagonal boundary, we see a new sort of biunitary vertex, with two adjacent shaded regions in its neighbourhood, and two adjacent unshaded regions. Previous results on biunitary connections tell us that such a vertex encodes the data of a \emph{quantum Latin square}\footnote{A quantum Latin square is an $n$-by-$n$ grid of elements of the Hilbert space $\mathbb{C}^n$, such that every row and column yields an orthonormal basis (see Definition~\ref{def:qls}.)}~\cite{musto_quantum_2016}, combinatorial objects of recent interest in quantum foundations. It follows that quantum Latin squares precisely characterise the dynamics of this clockwork-brickwork boundary.

Other interesting phenomena are possible. Here we begin at early times with a spatially homogeneous brickwork circuit, but at a certain spacetime point $P$ an operator inserts a clockwork region, the boundaries of which then propagate left and right along the causal light cone:

\begin{calign}
\label{fig:wedge_intro}
\begin{aligned}
\includegraphics[width=7cm]{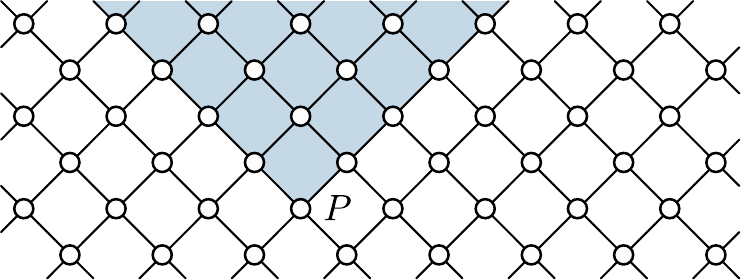}
\end{aligned}
&
\begin{aligned}
\includegraphics[width=7cm]{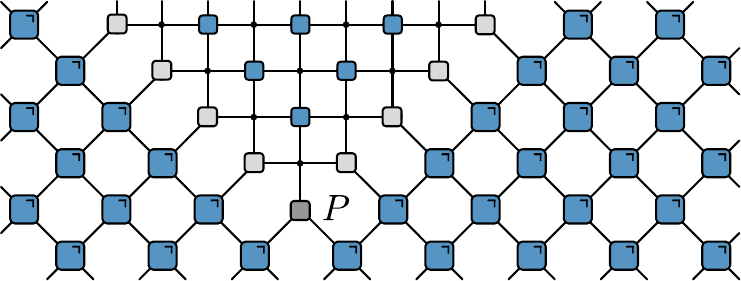}
\end{aligned}
\\ \nonumber
\textrm{\em (a) Shaded calculus representation}
&
\textrm{\em (b) Unitary circuit representation}
\end{calign}

\noindent
At point $P$ we see another shading pattern around the biunitary vertex, with one shaded and three unshaded regions. These biunitaries have been characterized as corresponding to \textit{unitary error bases}\footnote{A unitary error basis is family of unitary matrices which form an orthogonal basis of the operator space.}~\cite{vicary_higher_2012, Vicary_2012}, another important quantum combinatorial structure, introduced originally by Werner~\cite{werner_all_2001} for classifying quantum teleportation protocols.

A further scenario we are able to precisely characterize is the reflection of two incident boundaries with opposite velocities:

\begin{calign}
\label{fig:hadamard_intro}
\begin{aligned}
\includegraphics[width=7cm]{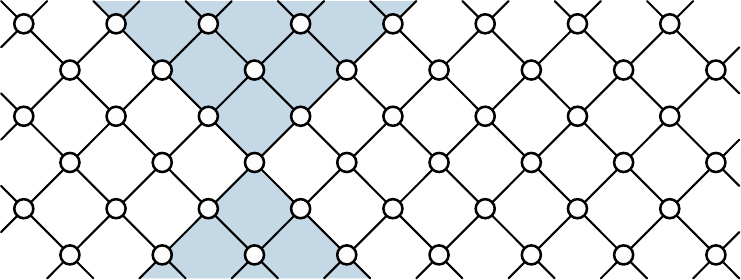}
\end{aligned}
&
\begin{aligned}
\includegraphics[width=7cm]{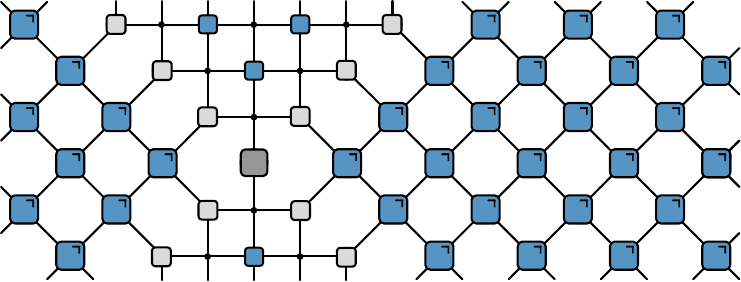}
\end{aligned}
\\ \nonumber
\textrm{\em (a) Shaded calculus representation}
&
\textrm{\em (b) Unitary circuit representation}
\end{calign}
Here the biunitarity property implies that the central reflection point is described by a complex Hadamard matrix, a result originally demonstrated by Jones in his work on subfactor theory~\cite{Jones1999}.

Another new possibility suggested by our approach is the construction of  homogenous homogeneous circuits with a regular shading pattern, such as the following, which makes further use of the Hadamard biunitary:

\begin{calign}
\label{fig:KIM_intro}
\begin{aligned}
\includegraphics[width=7cm]{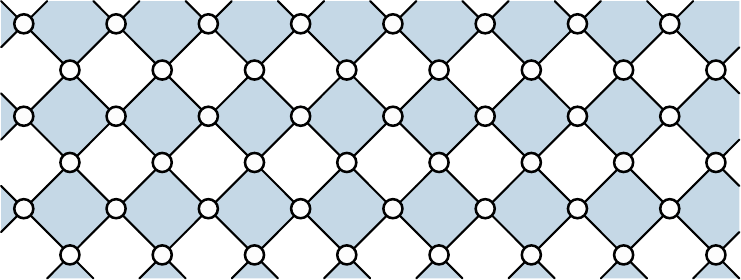}
\end{aligned}
&
\begin{aligned}
\includegraphics[width=7cm]{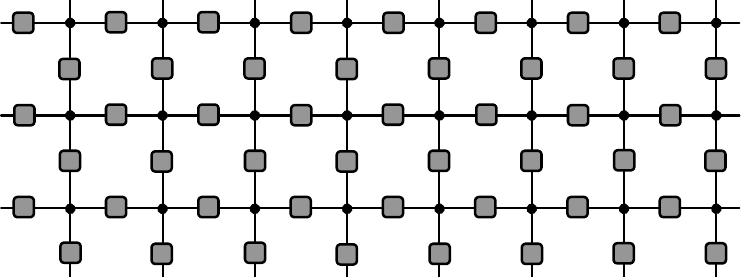}
\end{aligned}
\\ \nonumber
\textrm{\em (a) Shaded calculus representation}
&
\textrm{\em (b) Unitary circuit representation}
\end{calign}
Circuits of this form already appeared in the literature as \emph{ad hoc} decompositions of unitary circuits representing the `self-dual kicked Ising model' \cite{akila_particle-time_2016,bertini_entanglement_2019,gopalakrishnan_unitary_2019,ho_exact_2022,Stephen2022}. Using the shaded calculus, we see that this decomposition originates naturally from biunitarity, and the relationship to brickwork and clockwork circuits is made clear.

These phenomena can be combined in any spacetime orientation, to produce heterogeneous circuits with rich global structure, such as the following:

\begin{calign}
\label{fig:general}
\begin{aligned}
\includegraphics[width=7cm]{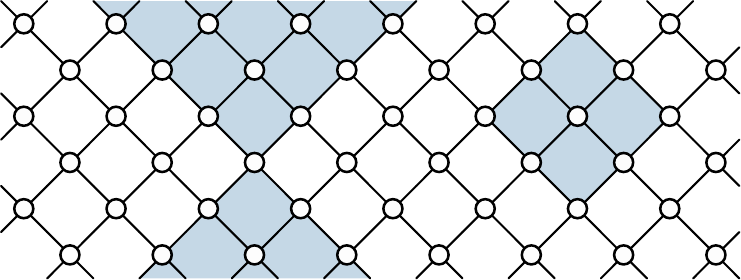}
\end{aligned}
&
\begin{aligned}
\includegraphics[width=7cm]{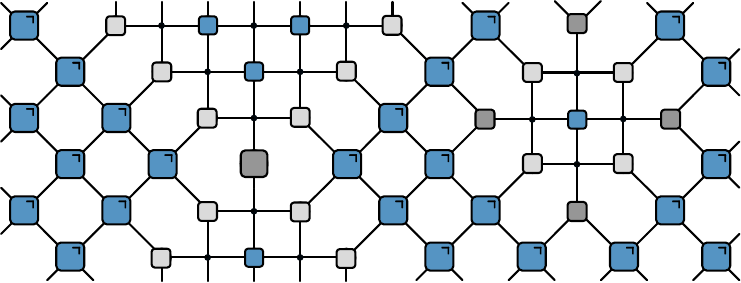}
\end{aligned}
\\ \nonumber
\textrm{\em (a) Shaded calculus representation}
&
\textrm{\em (b) Unitary circuit representation}
\end{calign}

Crucially, we show that all circuits produced by these techniques all satisfy a range of attractive properties relating to correlation functions and entanglement growth already known for the conventional brickwork and/or clockwork models. To achieve this, the key insight is a new  definition of solvable initial state appropriate for the shaded calculus, generalizing the standard notion of solvable matrix product state for conventional brickwork dual-unitary circuits. These solvable initial states allow us to derive exact results for entanglement dynamics of biunitary shaded circuits, and to prove that biunitary circuits exhibit exact thermalization after a finite number of time steps. In particular,  we obtain solvable initial states for clockwork circuits, where no such construction or results on entanglement growth were known so far.

\paragraph{Outline.} In \Cref{sec:biunitary} we introduce the shaded calculus and define biunitary brickwork circuits, focusing on the use of the graphical calculus to systematically and easily construct dual-unitary circuits, and present a dictionary between the shaded notation and the corresponding circuit representation. \Cref{sec:corr_entanglement} reviews relevant results from dual-unitary circuits on the dynamics of correlations, namely that single-site correlation functions vanish everywhere except on the light cone, which we show to hold more generally for biunitary circuits. In \Cref{sec:entanglement} we define solvable initial states for these biunitary circuits, showing that the resulting dynamics in discrete time results in exact thermalization after a finite number of time steps, and fully characterize these initial states for different circuits. In \Cref{sec:dimensions} we discuss the restrictions on the local Hilbert space dimensions enforced by heterogeneous circuit structures. \Cref{sec:conclude} presents outstanding questions and an outlook. In Appendix~\ref{app:param_KIM} we provide an explicit parameterisation of composite biunitaries arising in periodic heterogeneous structure.

\section{Shaded calculus and biunitarity}

\subsection{Introduction to the shaded calculus}
\label{sec:shaded}

In the graphical notation for linear algebra, tensors and tensor composites are represented in Penrose notation~\cite{penrose_penrose-applications--negative-dimensional-tensorspdf_1971}, as also popularized in the tensor network literature \cite{orus_practical_2014}. From the categorical perspective, the resulting diagrams can be interpreted as string diagrams for the pivotal monoidal category \textbf{Hilb} of finite-dimensional Hilbert spaces~\cite{Selinger_2010, heunen_2019}. In this notation, wires represent Hilbert spaces, and vertices represent linear maps between them, such that wiring diagrams represents composite linear maps. Diagrams such as the following are standard:
\begin{align}
\begin{aligned}
\includegraphics[width=0.25\textwidth]{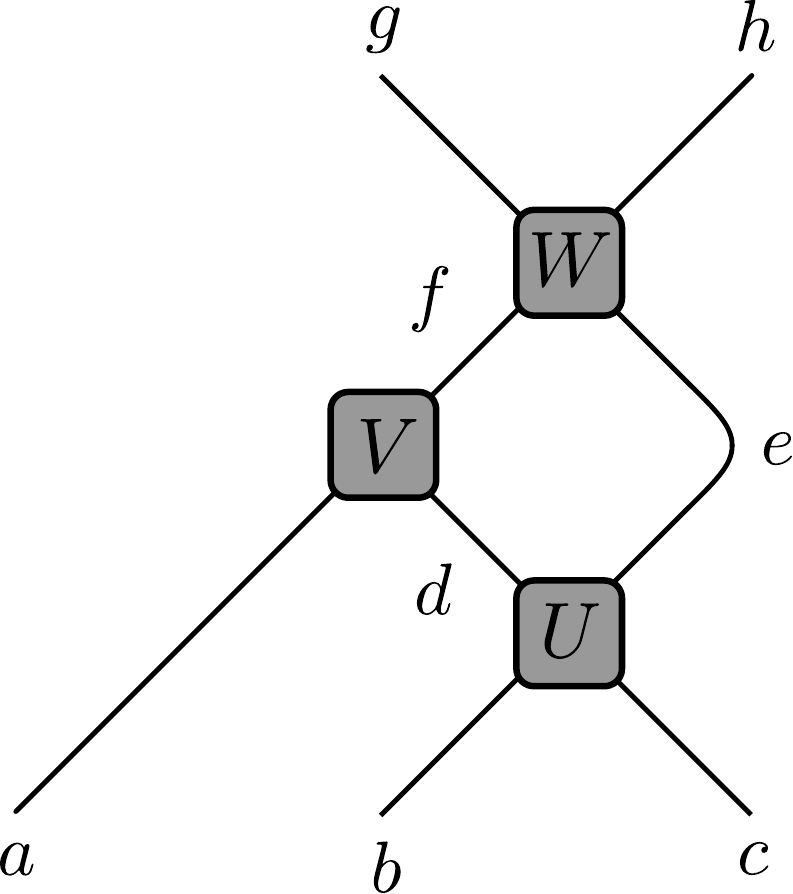}
\end{aligned}
\end{align}
This denotes a  composite of tensors $U$, $V$, $W$, with the indices $a$ through $h$ taken to represent basis elements of certain Hilbert spaces, whose dimensions are usually left implicit. Vertices represent complex numbers indexed by labels of adjacent wires; for example, the numbers $V_{a,d,f}$ represent the coefficients of the linear map $V$. There is an implicit sum over the index values of the closed wires $d,e,f$, and a direct sum over the remaining open wires, which gives the input and output Hilbert spaces of the composite.

Our work uses the shaded tensor calculus, which adds the feature of shaded regions. This is an instance of the planar graphical calculus for the dagger pivotal 2-category \textbf{2Hilb} of finite-dimensional 2--Hilbert spaces~\cite{Baez_1997, heunen_2019}, describing the theory of linear algebraic structures that can be composed in the plane, and also reflected and rotated (see introductions~\cite{Selinger_2010, heunen_2019} and research articles~\cite{hummon_surface_2012, bartlett_quasistrict_2014, schommer-pries_classification_2014, barrett_gray_2018, reutter_biunitary_2019}.) The shaded calculus is not as well-known as the traditional tensor notation, and so we introduce it carefully here. In an effort to demonstrate the simplicity of this notation, we choose a presentation that emphasizes its similarity to the traditional tensor calculus, avoiding explicit higher-categorical machinery.

An example of the shaded calculus is given in expression~\eqref{eq:shadedexample}(a), as follows:

\begin{calign}
\label{eq:shadedexample}
\begin{aligned}
\includegraphics[width=0.3\textwidth]{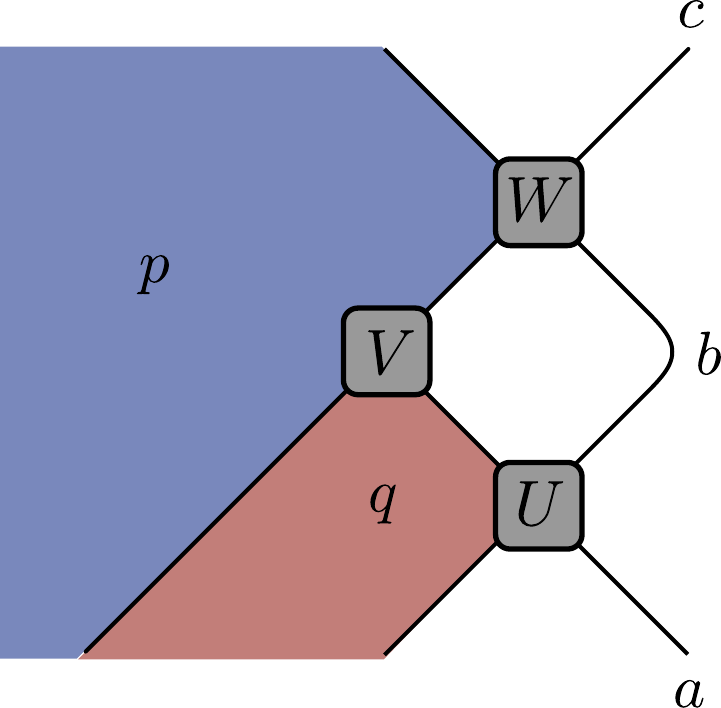}
\end{aligned}
&
\begin{aligned}
\includegraphics[width=0.1545\textwidth]{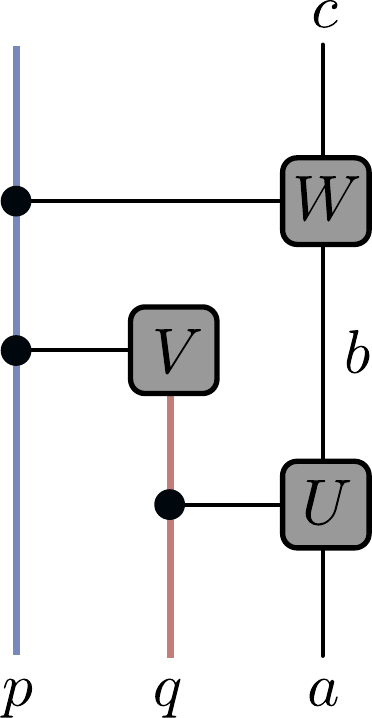}
\end{aligned}
\\ \nonumber
\textrm{\em (a) Shaded calculus diagram}
&
\textrm{\em (b) Controlled operator representation}
\end{calign}
This example has blue and red shaded regions, with parameters $p$ and $q$ respectively. These parameters are valued in some finite index set, determined by the dimension of the region, which we usually leave implicit. As with the conventional tensor notation discussed above, vertices represent complex numbers indexed by adjacent index labels, which can now include region indices as well as wire indices. The example above represents tensors $(U_{q})_{a,b}$, $V_{p,q}$ and $(W_p)_{b,c}$, where here we distinguish region indices by placing them within brackets, a notation which will be helpful later. Note that only bare wires carry an index; wires bordering shaded regions do not.\footnote{A natural generalisation of this calculus allows all wires to carry an index, including those wires bordering shaded regions, giving the full expressivity of the 2-category $\mathbf{2Hilb}$. However here we work entirely within a simpler fragment of the theory.} Again in common with the tensor calculus, we implicitly sum over parameters labelling closed regions, and take a direct sum over parameters labelling open regions.\footnote{Note that both regions in example \eqref{eq:shadedexample}(a) are open; we will see examples of closed regions below.} Unshaded regions can be considered to be trivially indexed over the 1-element set.

As a result, the data of diagram \eqref{eq:shadedexample}(a) can be equivalently expressed in traditional tensor notation as a family of \textit{controlled operators}, which we show as diagram~\eqref{eq:shadedexample}(b). Regions in the shaded diagram are now denoted as wires, which we draw here in the same colour, and vertices bordering the region are controlled by the corresponding wire. The control wires allow precisely the same indexing behaviour as encoded by the shaded diagram. However, note that when the shaded region is entirely within the input of a vertex, such as the region labelled $q$ in the input of $V$ above, then this is treated in  the circuit representation as an  input wire with no control; outputs behave similarly.

The benefits of the shaded calculus include the simplicity of the representation, and the stronger compositional properties, which we exploit heavily in this article. Of course, the traditional controlled tensor notation remains preferable in some instances, in particular where the connectivity is non-planar.

The present article studies the behaviour of 4-valent vertices with a variety of shading patterns. Using the techniques described above, any such shaded vertex can be directly represented as a controlled operator, and we will frequently make use of such a representation.

%--------------------------------------------------------------------------------------------------------------------------------------------------------
\subsection{Biunitary circuits}
\label{sec:biunitary}
In this Section, we review relevant results from Ref.~\cite{reutter_biunitary_2019} and define biunitarity for general vertices. 
We explicitly define the five biunitary building blocks of biunitary circuits: dual unitaries, unitary error bases, complex Hadamard matrices, quantum Latin squares, and quantum crosses. We then show how circuits of biunitary connections can be systematically constructed using the graphical calculus, in such a way that the resulting circuit is unitary along both the time and space direction.

\begin{defn}
A biunitary is defined as a vertex $U$
\begin{align}
\vcenter{\hbox{\includegraphics[width=0.15\linewidth]{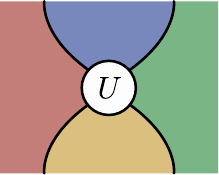}}}\,
\end{align}
with a conjugate vertex $U^{\dagger}$, such that these satisfy the ordinary notion of (vertical) unitarity
\begin{align}
&\vcenter{\hbox{\includegraphics[width=0.3\linewidth]{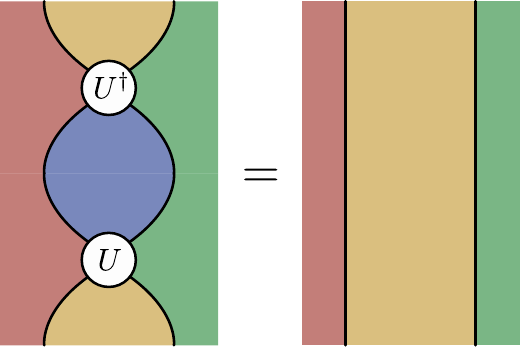}}}
&&
\vcenter{\hbox{\includegraphics[width=0.3\linewidth]{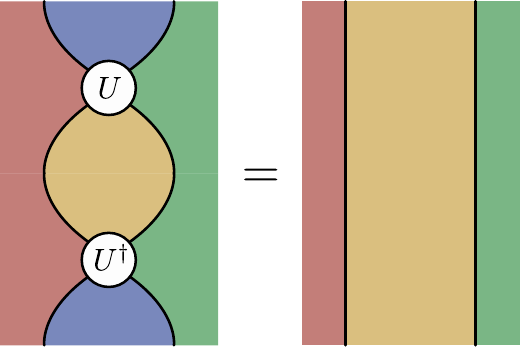}}}
\end{align}
as well as a notion of horizontal unitarity:
\begin{align}
&\vcenter{\hbox{\includegraphics[width=0.32\linewidth]{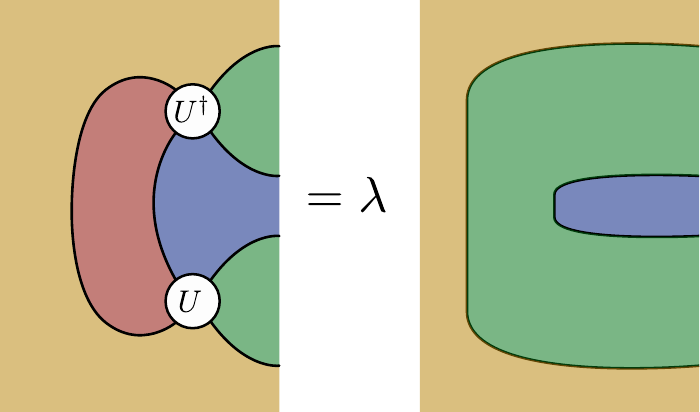}}}
&&
\vcenter{\hbox{\includegraphics[width=0.32\linewidth]{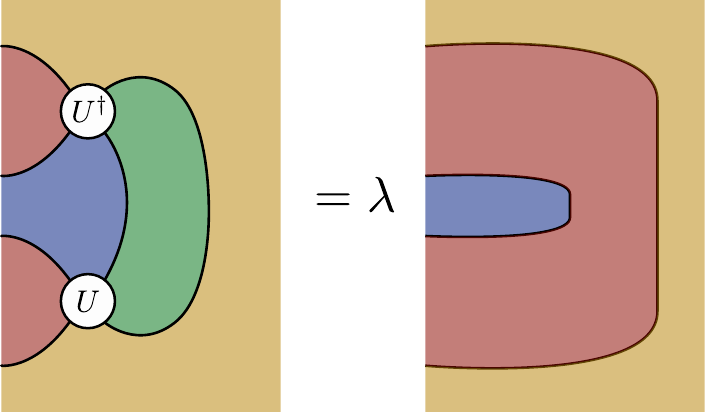}}}
\end{align}
The constant $\lambda$ can be uniquely determined from the biunitary and generally corresponds to a power of the local Hilbert space dimension $q$. We will typically omit such factors in our diagrams, such that all graphical equalities hold up to a scalar factor.
\end{defn}
%%
%Note that in this definition the different regions were marked with different colours, indicating that these regions can either be shaded or not -- leading to different quantum structures, that however all satisfy two notions of unitarity. The type of biunitary is defined by the specific shading pattern. Furthermore, every biunitary connection can be given an interpretation as a quantum gate, with the specific gate depending on the type of biunitary. 

\subsection{Quantum structures as biunitaries}
\label{sec:quantumstructures}

Here we unpack the definition of biunitarity for the different shading patterns and connections relevant for this work. In each case we describe the quantum combinatorial object that the biunitary corresponds to, and show explicitly how it is represented in controlled tensor notation.

\paragraph{Dual-unitary gates.}For a biunitary with no regions shaded, we obtain the ordinary tensor network concept of {dual-unitary gate}. These are two-site unitary matrices with matrix elements $U_{ab,cd}$ represented  as follows:\footnote{We denote matrix elements with commas and parentheses  where this helps to make the circuit representation more transparent.}
\begin{align}
U_{ab,cd} \,\,\,\,= \vcenter{\hbox{\includegraphics[width=0.1\linewidth]{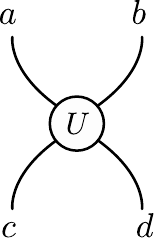}}}
\end{align}
For such a linear map to be biunitary, a necessary condition is that wires $c,b$ have the same dimension, and wires $a,d$ also have the same dimension. In this work we will make the simplifying assumption that all four wires have the same dimension when writing down explicit summations, but our results do not depend on this assumption. For a discussion of dual-unitary gates where the wires have different dimensions, we refer the reader to Ref.~\cite{Borsi2022}.

Vertical unitarity directly corresponds to unitarity of $U$, as follows:
\begin{align}
\vcenter{\hbox{\includegraphics[width=0.25\linewidth]{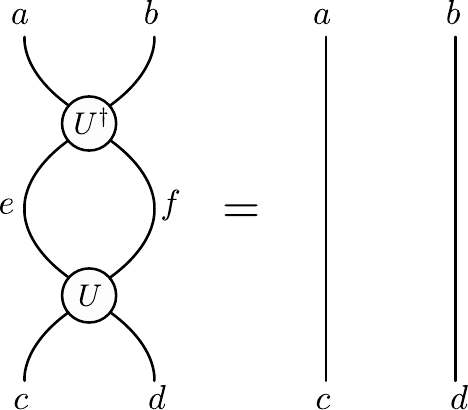}}} \qquad \Rightarrow \qquad \sum_{e,f=1}^q U^{\dagger}_{ab,ef} U_{ef,cd} = \delta_{ac}\delta_{bd}
\end{align}
Horizontal unitarity implies that  $\widetilde{U}_{ab,cd} := U_{bd,ac}$, the dual operator, is also unitary: 
\begin{align}
\vcenter{\hbox{\includegraphics[width=0.32\linewidth]{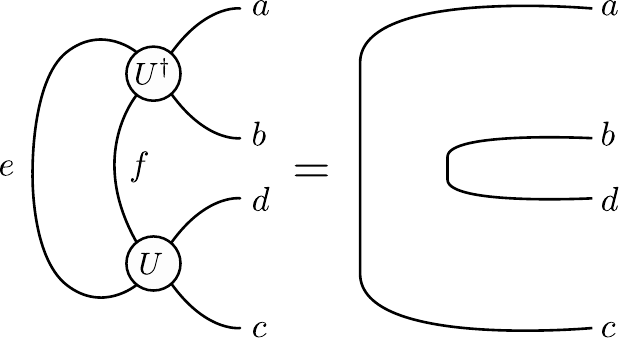}}}   \qquad \Rightarrow \qquad   \sum_{e,f=1}^q U^{\dagger}_{ea,fb} U_{ec,fd} =\sum_{e,f=1}^q \tilde{U}^{\dagger}_{ab,ef} \tilde{U}_{ef,cd} = \delta_{ac}\delta_{bd}
\end{align}

\noindent
\emph{Circuit representation.}
When representing biunitary circuits in terms of unitary gates, these vertices are denoted as two-site unitary gates and will be written as follows:
\begin{align}
 \vcenter{\hbox{\includegraphics[width=0.25\linewidth]{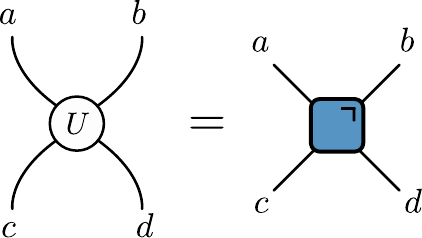}}}
\end{align}

\paragraph{Unitary error bases.}  Biunitary vertices with a single shaded region correspond to unitary error bases, important quantum algebraic structures defined as follows.%%
\begin{defn}
(Knill \cite{knill_non-binary_1996}) A \textit{unitary error basis} (UEB) is a complete orthogonal family of unitary matrices $\{U_a  \in U(H) | a=1 \dots q^2\}$, where orthogonality is with respect to the trace norm: $\Tr(U_a^{\dagger}U_b) = q \delta_{ab}$.
\end{defn}

\noindent
UEBs have found numerous applications in quantum information theory, discovered by Werner to underlie quantum teleportation \cite{werner_all_2001} as well as dense coding and error correction procedures~\cite{knill_non-binary_1996,shor_fault-tolerant_1996}. We characterise UEBs via biunitaries as follows.

\begin{lemma}[Vicary~\cite{Vicary_2012, vicary_higher_2012}]
Unitary error bases are precisely biunitaries with the following shading pattern:
\begin{align}
\label{shading_ueb}
(U_a)_{b,c} \quad = \quad \vcenter{\hbox{\includegraphics[width=0.12\linewidth]{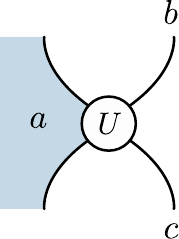}}}\,
\end{align}
\end{lemma}

\noindent
For such a vertex to be biunitary, a necessary condition is that the wires must carry the same dimension $q$, and the shaded region must carry the dimension $q^2$.

\begin{proof}
The vertical unitarity conditions correspond to each $U_{a}$ being a unitary matrix:
\begin{align}
\vcenter{\hbox{\includegraphics[width=0.25\linewidth]{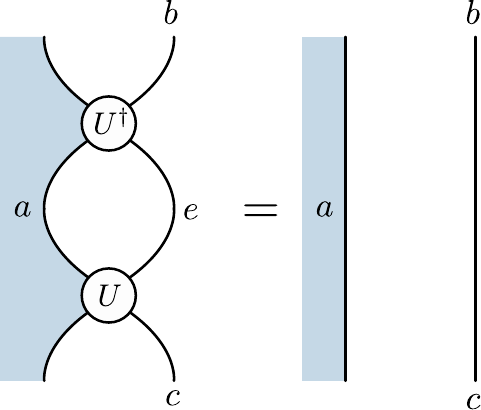}}} \qquad  \Rightarrow \qquad \sum_{e=1}^q (U_a)^{\dag}_{b,e} (U_a )_{e,c}   = \delta_{bc}\qquad \forall a=1\dots q^2\,,
\end{align}
The first horizontal unitarity condition guarantees the completeness of this basis, where we sum over the index corresponding to the closed region:
\begin{align}
\vcenter{\hbox{\includegraphics[width=0.32\linewidth]{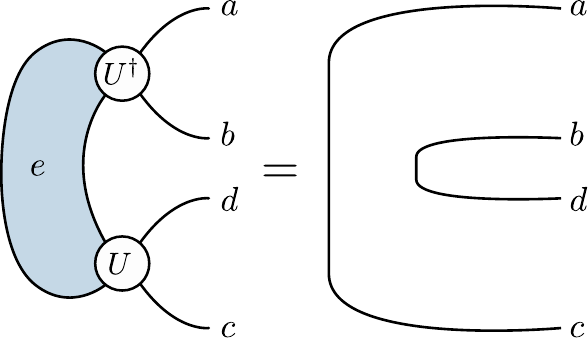}}} \qquad \Rightarrow \qquad \sum_{e=1}^{q^2} (U_e )^{\dagger}_{a,b} (U_e)_{d,c} = q\delta_{ac}\delta_{bd}
\end{align}
The second horizontal unitarity condition returns the trace-orthonormality of the different matrices, as follows:
\begin{align}
\vcenter{\hbox{\includegraphics[width=0.32\linewidth]{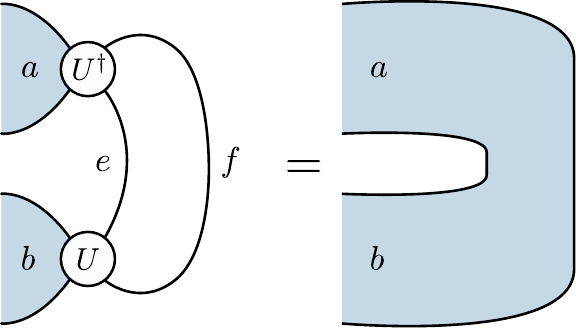}}}   \qquad  \Rightarrow \qquad  \sum_{e,f=1}^q ( U^{\dagger}_a )_{f,e} (U_b )_{e,f} = \Tr (U_a^{\dagger} U_b ) = q \delta_{ab}\,.
\end{align}
Alternatively, we can define $\tilde{U}_{a,bc} = (U_a)_{bc}$ as a linear map $\tilde{U} \in \mathbbm{C}^{q^2 \times q^2}$, and these conditions imply that $\tilde{U}$ is also unitary (note that $a=1\dots q^2$ whereas $b,c=1\dots q$).
\end{proof}

\vspace{5pt}
\noindent
\emph{Circuit representation.}
The representation of a  shaded calculus diagram as a controlled tensor network depends on the orientation of the shaded region. If the shaded region is on the left or right side we obtain a control wire, but not if it is above or below:
\begin{calign}
\vcenter{\hbox{\includegraphics[width=0.25\linewidth]{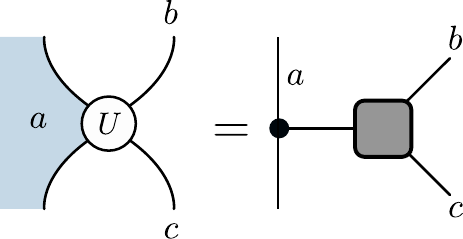}}}
&
\vcenter{\hbox{\includegraphics[width=0.25\linewidth]{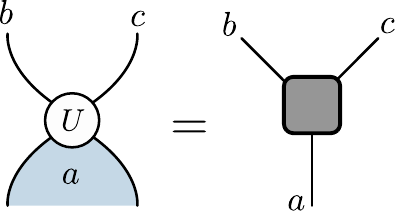}}}
\end{calign}

\paragraph{Hadamard matrices.} Here we investigate the definition of Hadamard structures via biunitaries. 
\begin{defn}
A \textit{Hadamard matrix} is a matrix $U \in \mathbbm{C}^{q \times q}$, which is proportional to a unitary matrix, such that every matrix entry has modulus~1. It follows that $U U^{\dagger} = U^{\dagger}U = q \mathbbm{1}$.
\end{defn}

\noindent
In the shaded calculus, they are represented as biunitaries with two opposite shaded regions.
\begin{lemma}[Jones~\cite{Jones1999}]
Hadamard matrices are precisely biunitaries with the following shading pattern:
\begin{align}
U_{a,b} \quad = \quad \vcenter{\hbox{\includegraphics[width=0.12\linewidth]{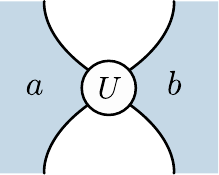}}}\,
\end{align}
\end{lemma}

\noindent
For such a vertex to be biunitary, a necessary  condition is that the regions $a$ and $b$ must have the same dimension.

\begin{proof}
Vertical unitarity fixes all matrix elements to have modulus one:
\begin{align}
\vcenter{\hbox{\includegraphics[width=0.25\linewidth]{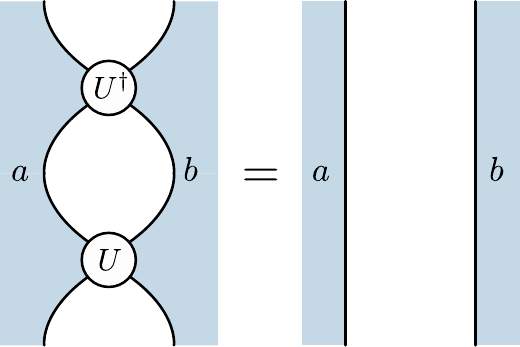}}} \qquad  \Rightarrow \qquad U_{a,b} (U_{a,b})^{\dagger} = |U_{a,b}|^2 = 1,\qquad \forall a,b=1\dots q\,,
\end{align}
Horizontal unitarity fixes these matrices to be proportional to unitary matrices, where now $\lambda = q$: 
\begin{align}
\vcenter{\hbox{\includegraphics[width=0.32\linewidth]{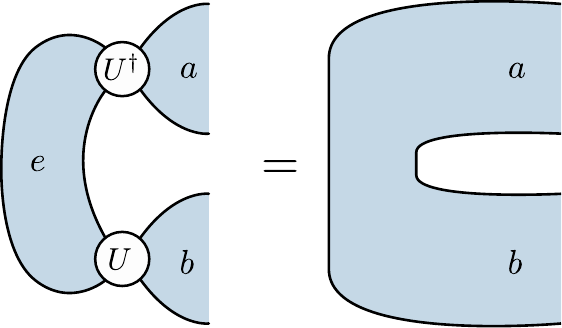}}} \qquad \Rightarrow \qquad \sum_{e=1}^{q} (U_{e,a} )^{\dagger} (U_{e,b}) = q\delta_{ab}\,,
\end{align}
This completes the proof.
\end{proof}

\vspace{5pt}
\noindent
\emph{Circuit representation.}
In the language of quantum circuits, a complex Hadamard matrix with vertical shaded regions corresponds to a 2-controlled phase $U_{a,b}$ with $a,b$ acting as control parameters. A complex Hadamard matrix with horizontal shaded regions corresponds to a one-site unitary gate $U / \sqrt{q}$, and both are represented by
\begin{calign}
\vcenter{\hbox{\includegraphics[width=0.33\linewidth]{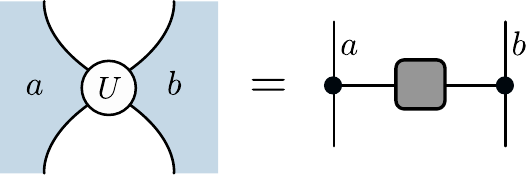}}}
&
\vcenter{\hbox{\includegraphics[width=0.184\linewidth]{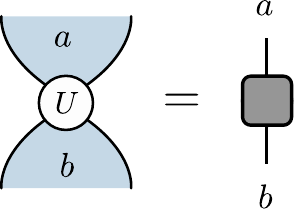}}}
\end{calign}

\paragraph{Quantum Latin squares.} Biunitaries can also be used to give a definition of quantum Latin square, as we now explore.

\begin{defn}
\label{def:qls}
(Musto \& Vicary \cite{musto_quantum_2016}) For a $q$-dimensional Hilbert space $H$, a \textit{quantum Latin square} (QLS) is a square grid of states  $\{Q_{a,b}\in H\,|a,b=1\dots q\}$, such that each row and column yields an orthonormal basis.
\end{defn}

\begin{lemma}[Jones~\cite{Jones1999}]
Quantum Latin squares are precisely biunitaries with two adjacent shaded regions:\begin{align}
(U_{a,b})_{c} \quad=\quad \vcenter{\hbox{\includegraphics[width=0.12\linewidth]{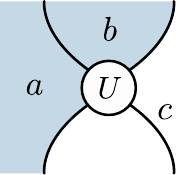}}}
\end{align}
\end{lemma}

\noindent
For such a vertex to be biunitary, a necessary condition is that the wire and both regions all have the same dimension $q$.

\begin{proof}
The first vertical unitarity equation gives completeness of the elements of each row:
\begin{align}
\vcenter{\hbox{\includegraphics[width=0.25\linewidth]{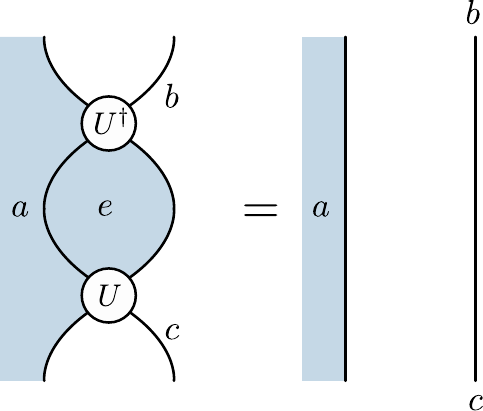}}} \qquad  \Rightarrow \qquad \sum_{e=1}^q (U_{a,e})^{\dagger}_{b} (U_{a,e} )_{c}  = \delta_{bc},\qquad \forall a=1\dots q\,
\end{align}
The second vertical unitarity returns the orthonormality of the elements of each row:
\begin{align}
\vcenter{\hbox{\includegraphics[width=0.25\linewidth]{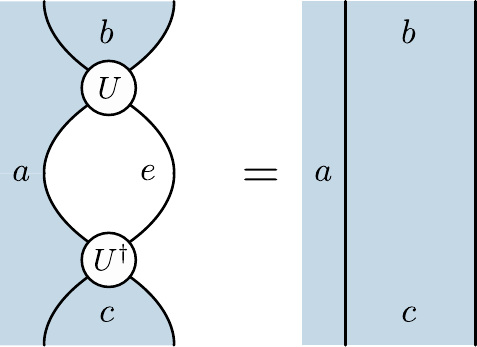}}} \qquad  \Rightarrow \qquad \sum_{e=1}^q  (U_{a,b} )_{e} (U_{a,c})^{\dagger}_{e}  = \delta_{bc},\qquad \forall a=1\dots q\,
\end{align}
Horizontal unitarity leads to the corresponding relation for the columns of $U$.
\end{proof}

\vspace{5pt}
\noindent
\emph{Circuit representation.}
When representing a  quantum Latin square as a circuit element, the control wire corresponds to the shaded region which is oriented to the side. The possible circuit elements are therefore as follows, depending on the orientation of the biunitary:
\begin{calign}
 \vcenter{\hbox{\includegraphics[width=0.25\linewidth]{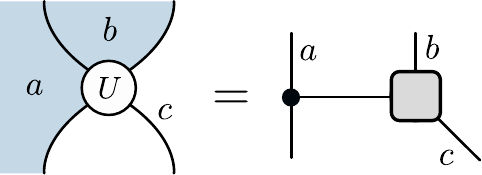}}}
&
 \vcenter{\hbox{\includegraphics[width=0.25\linewidth]{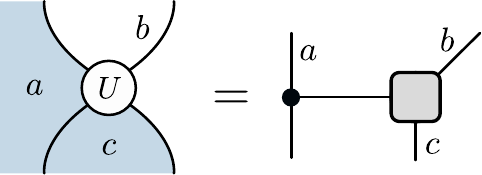}}}
\\[5pt]
 \vcenter{\hbox{\includegraphics[width=0.25\linewidth]{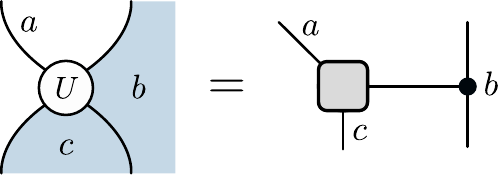}}}
 &
 \vcenter{\hbox{\includegraphics[width=0.25\linewidth]{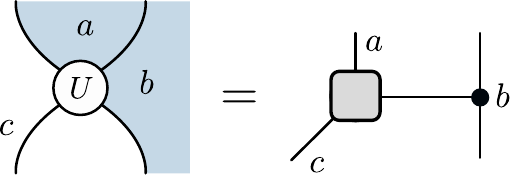}}}
\end{calign}

\paragraph{Quantum crosses.} The final combinatorial structure we will consider is the quantum cross.
\begin{defn}
\label{def:quantumcross}
On a $q$-dimensional Hilbert space $H$, a \textit{quantum cross} is a collection of $q^2$ unitary matrices $\{U_{a,c}  \in U(H) | a,c=1 \dots q\}$ with matrix elements  $(U_{a,c})_{b,d}$, such that the matrices $\tilde{U}_{b,d}$ with matrix elements $(\tilde{U}_{b,d})_{a,c} = (U_{a,c})_{b,d}$ are also unitary.
\end{defn}

\noindent
Quantum crosses correspond precisely to the DUIRF gates of Prosen~\cite[Section~III.C]{prosen_many-body_2021}, and we can characterise them in terms of biunitaries as follows.

\begin{lemma}
Quantum crosses are precisely biunitaries with all regions shaded:
\begin{align}
(U_{a,c})_{b,d} \quad=\quad \vcenter{\hbox{\includegraphics[width=0.12\linewidth]{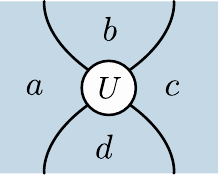}}}\,
\end{align}
\end{lemma}

\noindent
For such a vertex to be biunitary, a necessary condition is for opposite regions to have the same dimension.

\begin{proof}
Vertical unitarity corresponds to ordinary unitarity of the matrices $U_{a,b}$:
\begin{align}
\vcenter{\hbox{\includegraphics[width=0.25\linewidth]{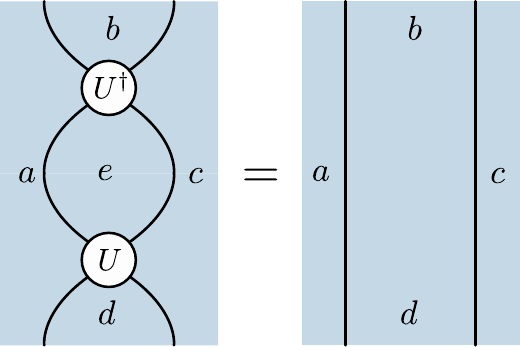}}} \qquad  \Rightarrow \qquad \sum_{e=1}^q (U_{a,c})^{\dagger}_{b,e} (U_{a,c})_{e,d}  = \delta_{bd},\qquad \forall a,c=1\dots q\,
\end{align}
Horizontal unitarity corresponds to unitarity of the matrices $\tilde{U}$ given in Definition~\ref{def:quantumcross}:
\begin{align}
\vcenter{\hbox{\includegraphics[width=0.32\linewidth]{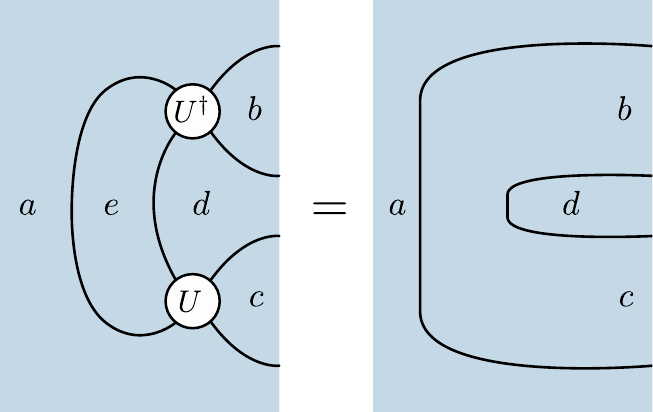}}}   \,  \Rightarrow \,  \sum_{e=1}^q (U_{e,b})^{\dagger}_{a,d}(U_{e,c})_{d,a} &= \sum_{e=1}^q (\tilde{U}_{d,a})^{\dagger}_{b,e} (\tilde{U}_{d,a} )_{e,c}= \delta_{bc}\,, \nonumber\\
& \forall a,d=1\dots q.
\end{align}
This completes the proof.
\end{proof}

\vspace{5pt}
\noindent
\emph{Circuit representation.}
When interpreting quantum crosses as circuit components, we obtain single-site double-controlled unitary gates where the 2 outer wires correspond to control parameters, i.e.
\begin{align}
 \vcenter{\hbox{\includegraphics[width=0.3\linewidth]{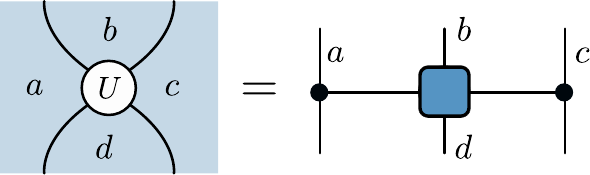}}}
\end{align}

\paragraph{Three shaded regions.} The above examples fully exhaust all biunitary connections: either no shaded regions, a single shaded region, or four shaded regions. The only remaining option would be a biunitary with three shaded regions, having the following graphical form:
\begin{align}
(U_{a,c})_{b}
\quad=\quad
\vcenter{\hbox{\includegraphics[width=0.12\linewidth]{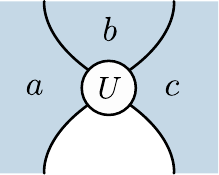}}}\,
\end{align}
However, the vertical unitarity condition implies that the the region labelled $b$ has dimension 1, and so this region is in fact unshaded, and we recover the Hadamard shading pattern.

\paragraph{Completeness.} To conclude this section, we note that all possible biunitaries for local Hilbert space dimension $q=2$ are fully known. In this case a complete parameterization was given for dual-unitary gates by Bertini \emph{et al.}~\cite{bertini_exact_2019}, and for quantum crosses by Prosen~\cite{prosen_many-body_2021}. For $q=2$ the unitary error bases are equivalent to the Pauli basis \cite{klappenecker_unitary_2003}, and complex Hadamard matrices are equivalent to the Fourier matrices~\cite{tadej_concise_2006}. In a $2 \times 2$ quantum Latin square a single vector can be freely chosen, after which orthonormality fixes all other vectors up to phase. 

For Hilbert space dimensions larger than 2 complete parametrizations of general biunitary connections remain absent in all cases. 
However, biunitaries for a larger Hilbert space can be systematically constructed: the diagonal composition of biunitaries is again biunitary, such that it is possible to compose biunitaries with a specific shading pattern out of other biunitaries which may have a different shading pattern. Such constructions are the main topic of Ref.~\cite{reutter_biunitary_2019}, and were already used in Ref.~\cite{claeys_emergent_2022} to construct specific classes of dual-unitary gates.

\subsection{Biunitary circuits}
The introduction of these different biunitaries now allows us to present the main result of this work. Dual-unitary `brickwork' circuits are circuits of dual-unitary gates, which we take to be arranged with the specified space-time indexing:
\begin{align}
\label{fig:indexed}
\vcenter{\hbox{\includegraphics[width=0.6\linewidth]{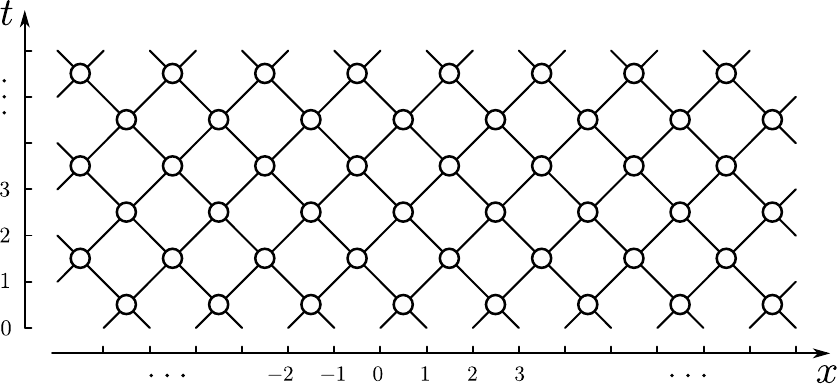}}}  
\end{align}
These circuits are unitary when acting along the vertical direction and acting along the horizontal direction, which has drastic consequences when we interpret these circuits as describing unitary dynamics of a discrete lattice $x$ (horizontal direction) along a discrete time $t$ (vertical direction). Every layer of dual-unitary gates can be interpreted as describing a single time step, and the resulting dynamics is exactly solvable \cite{bertini_exact_2019}. We now generalize this notion to biunitary circuits.

\begin{defn}
A \textit{biunitary circuit} is a brickwork circuit constructed from biunitary vertices.
\end{defn}

For convenience, we represent these biunitary circuits with a transparent gray background, with the implicit assumption that every vertex is biunitary.
\begin{align}
\label{fig:def_biunitarycircuit}
\vcenter{\hbox{\includegraphics[width=0.6\linewidth]{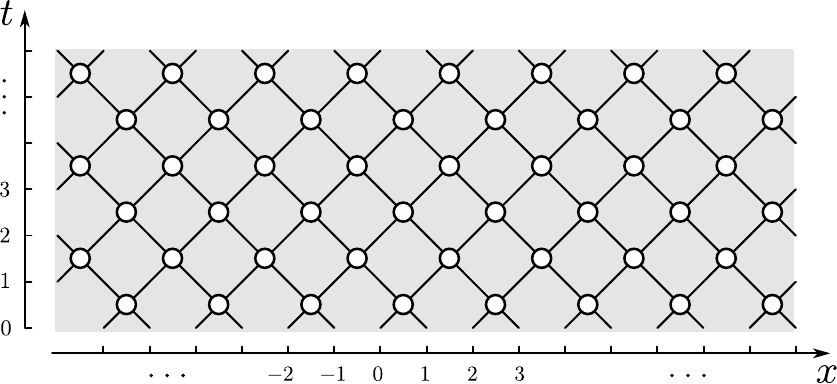}}}  
\end{align}
In practice, biunitary circuits can be directly constructed by taking a brickwork diagram and introducing shading patterns following a tiling pattern consistent with biunitarity. All boundaries between shaded and unshaded regions must correspond to one of the presented biunitary connections, which prohibits the appearance of vertices bordering on three shaded regions. We emphasize that the transparent gray background can indicate both a shaded or an unshaded region, with the only restriction being that all vertices are biunitary.

These circuits are again unitary both along the vertical and horizontal direction. This property is a direct consequence of the fact that arbitrary finite diagonal composites of biunitaries are again biunitary~\cite{reutter_biunitary_2019}. The described shading patterns guarantee that all biunitaries are arranged in such a way that all compositions are diagonal compositions, such that the total circuit is biunitary and can be seen as a dual-unitary transformation. This property is immediately clear from the graphical notation: biunitary circuits represent unitary transformations by construction, as implied by the fact that every biunitary connection can be represented as a (controlled) unitary, and rotating the circuits by 90$^{\circ}$ returns an equally valid biunitary circuit, which is hence also unitary.

We illustrate some possible shading patterns below, combined with their representation as a quantum circuit consisting of (controlled) unitary gates. A fully unshaded circuit corresponds to a brickwork circuit of dual-unitary gates. A fully shaded circuit corresponds to a circuit composed of 2-controlled 1-qubit gates, returning the previously studied dual-unitary interactions round-a-face (clockwork) circuits:
\begin{calign}
\begin{aligned}
\includegraphics[width=7cm]{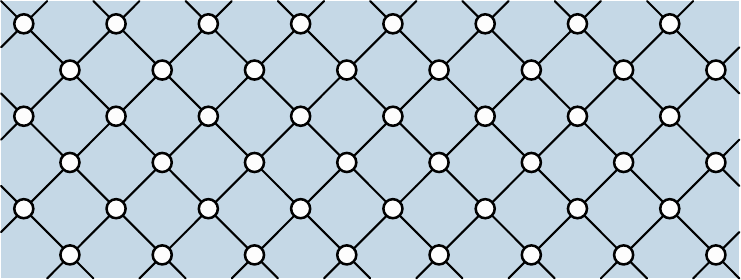}
\end{aligned}
&
\begin{aligned}
\includegraphics[width=7cm]{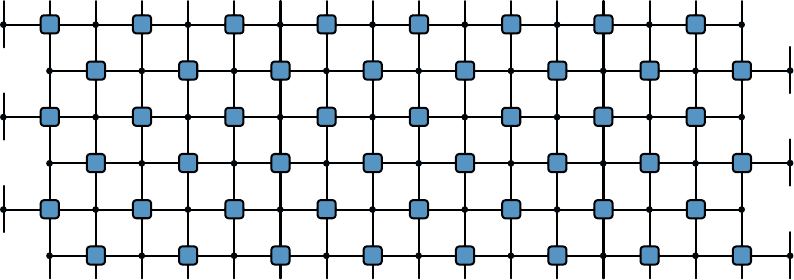}
\end{aligned}\quad
\end{calign}
Dual unitary clockwork and brickwork circuits can be `glued' together by including a diagonal boundary of quantum Latin squares, here corresponding to a boundary of 1-controlled 1-site gates:
\begin{calign}
\begin{aligned}
\includegraphics[width=7cm]{fig_diagonal_a}
\end{aligned}
&
\begin{aligned}
\includegraphics[width=7cm]{fig_diagonal_b}
\end{aligned}\quad
\end{calign}
The intersection of two such boundaries corresponds to a unitary error basis. If the shaded region lies along the top or bottom the UEB corresponds to a unitary gate mapping two $q$-dimensional wires to a single $q^2$-dimensional wire:
\begin{calign}
\begin{aligned}
\includegraphics[width=7cm]{fig_wedge_a}
\end{aligned}
&
\begin{aligned}
\includegraphics[width=7cm]{fig_wedge_b}
\end{aligned}\quad
\end{calign}
If the shaded region of the unitary basis lies along the left or right side it corresponds to a 1-controlled 1-site unitary:

\begin{calign}
\begin{aligned}
\includegraphics[width=7cm]{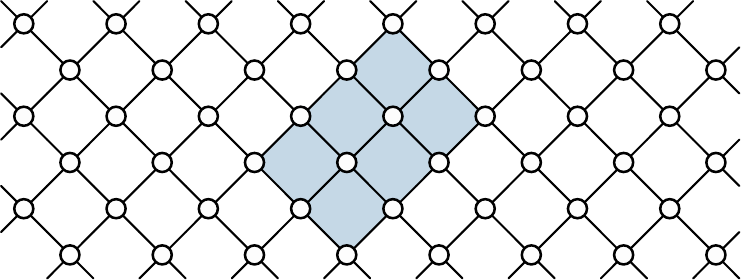}
\end{aligned}
&
\begin{aligned}
\includegraphics[width=7cm]{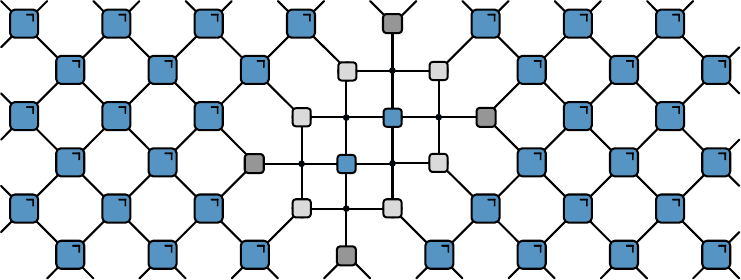}
\end{aligned}\quad
\end{calign}
Note that purely horizontal wires correspond to `control' parameters for the bordering unitaries, whereas all wires with a vertical components are either the input or output of a (possibly controlled) unitary gate. Both in the shaded calculus and in the unitary gate representation biunitary circuits are clearly mapped to biunitary circuits when exchanging the role of time and space, i.e. when rotating the full circuit by 90$^{\circ}$.

\subsection{Periodic biunitary circuits}
\label{subsec:periodic}
As one final example we consider a circuit with a regular shading pattern, periodic in both time and space. We choose the simplest example of a checkerboard pattern of shaded regions:
\begin{calign}
\label{fig:KIM}
\begin{aligned}
\includegraphics[width=7cm]{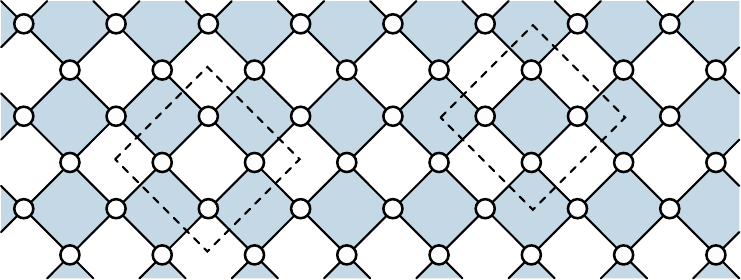}
\end{aligned}
&
\begin{aligned}
\includegraphics[width=7cm]{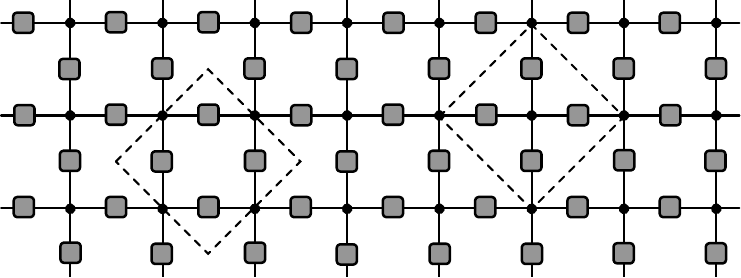}
\end{aligned}
\\ \nonumber
\textrm{\em (a) Shaded calculus representation}
&
\textrm{\em (b) Unitary circuit representation}
\end{calign}
The resulting unitary circuit is composed purely out of complex Hadamard matrices. Due to the periodicity of the circuit a unit cell can be defined, and in both sides of Eq.~\eqref{fig:KIM} two different choices of unit cell have been marked, each consisting of four complex Hadamard matrices.

Depending on the choice of unit cell these circuits can be reinterpreted as either brickwork or clockwork circuits, constructed out of composite dual-unitary gates or composite quantum crosses respectively. Crucially, all these composite objects are guaranteed to be biunitary since diagonal composites of biunitaries are again biunitary~\cite{reutter_biunitary_2019}. The shading pattern of the composite object fixes the kind of biunitary, here dual-unitary gates and quantum crosses respectively. In the first case, the four complex Hadamard matrices in the unit cell can be grouped together as a dual-unitary gate, as follows:\footnote{This construction of dual-unitary gates out of complex Hadamard matrices also appears in Ref.~\cite{claeys_emergent_2022}, generalizing previous constructions from Refs.~\cite{gutkin_exact_2020,aravinda_dual-unitary_2021,Borsi2022}.}
\begin{align}\label{eq:KIM_DU}
\vcenter{\hbox{\includegraphics[width=0.45\linewidth]{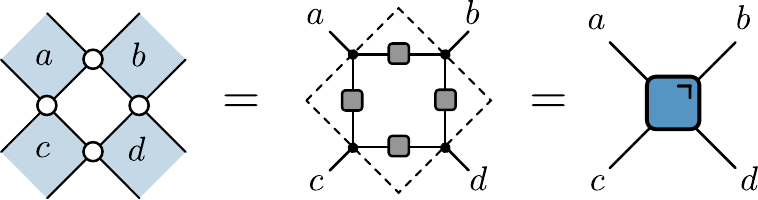}}}\,
\end{align}
In the second case, the four complex Hadamard matrices compose into a quantum cross:
\begin{align}\label{eq:KIM_cross}
\vcenter{\hbox{\includegraphics[width=0.5\linewidth]{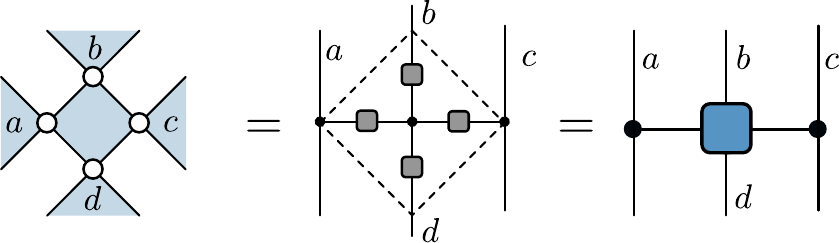}}}\,
\end{align}
Here we have extended the notation of black circles to correspond to tensors which are only nonzero if all connecting wires agree:
\begin{calign}
\vcenter{\hbox{\includegraphics[width=0.18\linewidth]{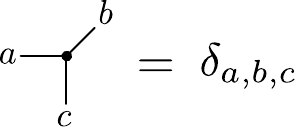}}}
&
\vcenter{\hbox{\includegraphics[width=0.2\linewidth]{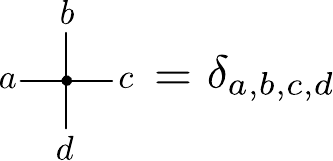}}}
\end{calign}
Note that this notation is consistent with the use of black circles to indicate connections between control wires and controlled unitary gates.

Circuits of the above form are prevalent in the literature on dual-unitary gates, where they appear as decompositions of the self-dual kicked Ising model \cite{akila_particle-time_2016,bertini_entanglement_2019,gopalakrishnan_unitary_2019,ho_exact_2022,Stephen2022}. This decomposition is typically argued from the explicit parametrization of the dual-unitary gate representing the kicked Ising model, whereas here it is directly apparent from the underlying biunitary circuit construction. Explicit parametrizations are given in Appendix~\ref{app:param_KIM} for completeness.

\section{Light-cone correlations}
\label{sec:corr_entanglement}
One of the main characteristics of dual-unitary circuits is that all ultralocal correlation coefficients vanish everywhere except on the edge of a causal light cone, where they can be efficiently calculated. This property has led to dual-unitary circuits being termed `exactly solvable'. The proof for dual-unitary gates is purely graphical and depends on the horizontal and vertical unitarity of the building blocks~\cite{bertini_exact_2019}. The introduction of the shaded calculus allows this proof to be extended to general biunitary circuits.

\begin{defn}[Operator at a given site] For an operator $\rho \in \mathbbm{C}^{q \times q}$, and some choice of site $x \in \mathbb N$ according to the indexing convention of \eqref{fig:indexed}, we define $\rho(x)$ as the operator $\rho$ acting at site $x$, as follows:
\begin{align}
\rho(x) &= \mathbbm{1} \otimes \dots \otimes \mathbbm{1} \otimes \underbrace{\rho}_{x} \otimes \mathbbm{1} \otimes \dots \otimes \mathbbm{1}, \nonumber\\
&= \vcenter{\hbox{\includegraphics[width=0.6\linewidth]{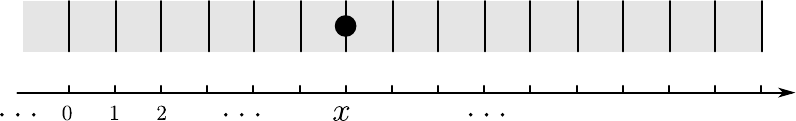}}}
\end{align}
\end{defn}

\noindent
In the second line the black circle represents $\rho$ and we have used the convention from Fig.~\ref{fig:def_biunitarycircuit}.

Any biunitary circuit represents a unitary operator $\mathcal{U}$, with the number of rows in the biunitary circuit equal to the number of discrete time steps. From now on, we will introduce the number of time steps as an additional label $t$ and write $\mathcal{U}_t$.

\begin{defn}
At integer coordinates $x,t$, we define the \textit{correlation function} as follows:
\begin{align}
c_{\rho \sigma}(x,t) \equiv \braket{\,\mathcal{U}_t^{\dagger}\rho(0)\, \mathcal{U}_t\, \sigma(x)} \equiv \mathrm{Tr}\big[\mathcal{U}_t^{\dagger}\rho(0)\, \mathcal{U}_t\, \sigma(x) \big] /\, \mathrm{Tr}\left[\mathbbm{1}\right]
\end{align}
The trace indicates that these correlation functions are taken w.r.t. the maximally mixed state, also known as the infinite-temperature state, $\braket{O} = \mathrm{Tr}(O)/\mathrm{Tr}(\mathbbm{1})$.
\end{defn}
Note that we here assume that the wires on sites $0$ and $x$ do not border a shaded region and represent a free index, but the case where $0$ and/or $x$ correspond to shaded region is analogous.
Without loss of generality we take both $\rho$ and $\sigma$ to be traceless. Otherwise we can redefine $\rho \to \rho-\mathrm{Tr}(\rho) /\mathrm{Tr}(\mathbbm{1}) \times \mathbbm{1}$ and similar for $\sigma$, and use the linearity of the correlation function in $\rho$ and that the correlation functions for the identity $\mathbbm{1}(x) = \mathbbm{1}$ are trivially constant since $\mathcal{U}_t^{\dagger}\mathbbm{1}(x)\, \mathcal{U}_t =\mathcal{U}_t^{\dagger}\, \mathcal{U}_t = \mathbbm{1}$.
\begin{theorem}
The infinite-temperature correlation functions vanish everywhere except on the edge of the causal light cone: $c_{\rho \sigma}(x,t) = 0$ unless $|x| = t,t-1.$
\end{theorem}

\begin{proof}
The proof proceeds graphically using the shaded calculus. 
In order to represent the biunitary circuits in full generality, we again represent $\mathcal{U}_t$ as a biunitary circuit with a transparent gray background, whereas we represent $\mathcal{U}_t^{\dagger}$ as the corresponding circuit with a transparent red background, again with the convention that every region in the above circuit can either be shaded or not and assuming that every vertex is biunitary. Graphically,
\begin{align}
\mathcal{U}_t^{\dagger}\rho(0)\, \mathcal{U}_t \,\,\,=\,\,\, \vcenter{\hbox{\includegraphics[width=0.4\linewidth]{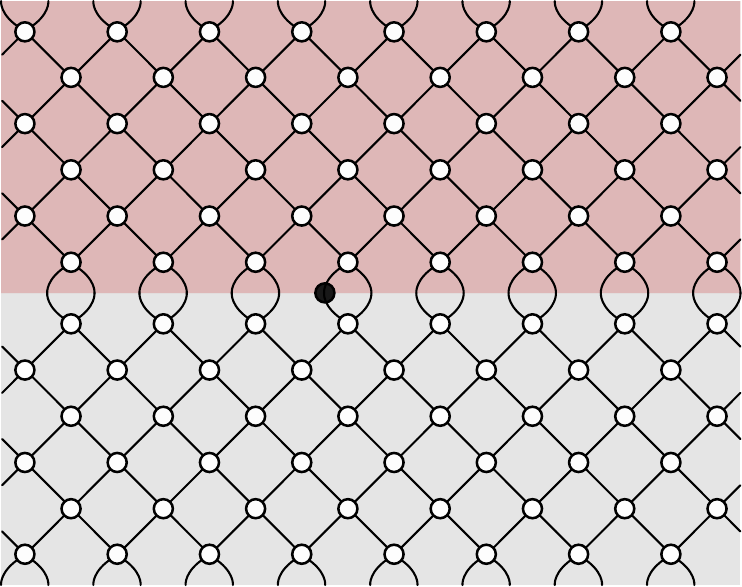}}}\,.
\end{align}
Here we have represented $\rho$ as a black circle. In these circuits we will only ever contract each biunitary connection $U$ with the corresponding $U^{\dagger}$, so we have made the labels implicit.
Through the repeated use of vertical unitarity, the above circuit can be simplified by first eliminating pairs of biunitary connections $U$ and $U^{\dagger}$ in the first row, and subsequently repeating this simplification, to yield the following:
\begin{align}
\vcenter{\hbox{\includegraphics[width=0.4\linewidth]{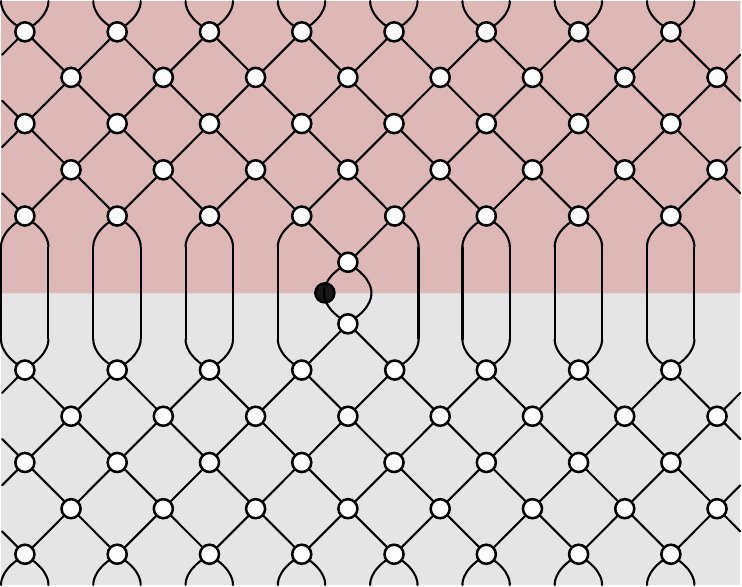}}} \,\,\,=\,\,\, \vcenter{\hbox{\includegraphics[width=0.4\linewidth]{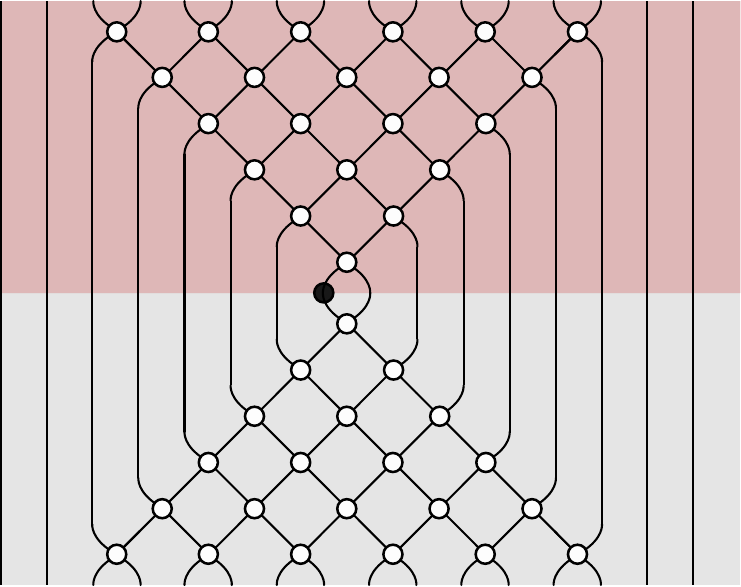}}}\,
\end{align}
This operator acts nontrivially only within the `hourglass' shape and acts as the identity everywhere else. The support grows linearly in the number of time steps, representing the causal structure of unitary brickwork circuits. Note that this property only depends on the vertical unitarity. For $|x|>t$ the argument immediately follows from the above representation of $\mathcal{U}_t^{\dagger}\rho(0)\mathcal{U}_t$, which acts as the identity on the support of $\sigma(x)$, and hence:
\begin{align}
    c_{\rho \sigma}(|x|>t,t) \equiv  \Tr\big[\mathcal{U}_t^{\dagger}\rho(0)\,\mathcal{U}_t \sigma(x)\big] = \Tr\big[\mathcal{U}_t^{\dagger}\rho(0)\,\mathcal{U}_t\big] \Tr\left[\sigma(x)\right] = \Tr(\rho) \Tr(\sigma) = 0\,.
\end{align}

For $|x|<t-1$ the correlations vanish due to the additional horizontal unitarity. 
The trace is represented by dashed lines at the top of the bottom, where we assume periodic boundary conditions connecting the wires and regions at the top and bottom of the diagram. Representing $\sigma$ as a black circle with support inside the causal light cone of $\rho$, a single application of horizontal unitarity can be used to remove the rightmost biunitaries and simplify the corresponding diagrams as follows:
\begin{align}
\vcenter{\hbox{\includegraphics[width=0.4\linewidth]{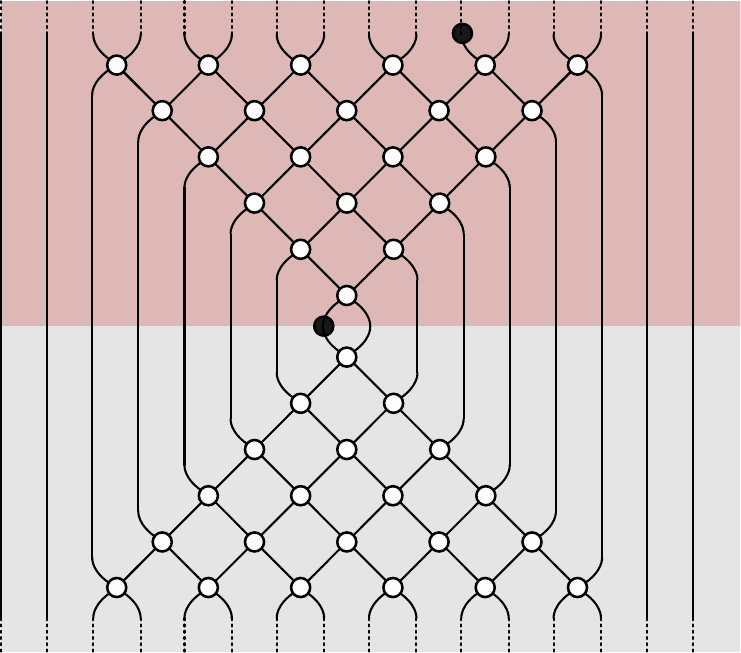}}} \,\,\,=\,\,\, \vcenter{\hbox{\includegraphics[width=0.4\linewidth]{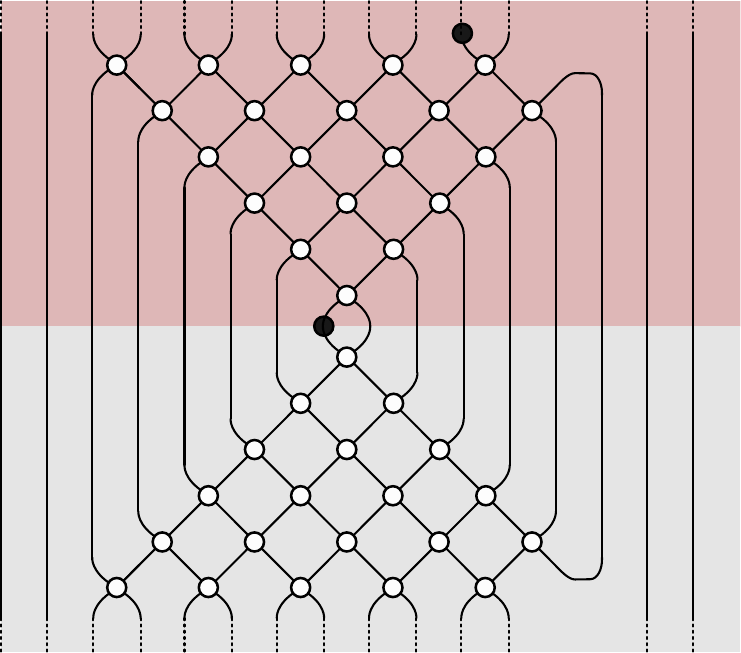}}}\,
\end{align}
These diagrams can be further `telescoped' using horizontal unitarity until the initial operator $\rho$ is encountered:
\begin{align}
\vcenter{\hbox{\includegraphics[width=0.4\linewidth]{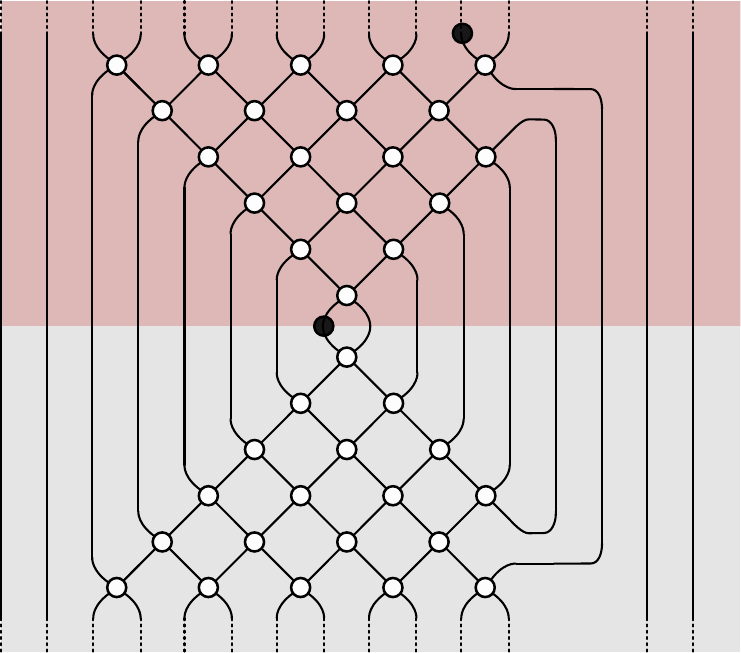}}} \,\,\,=\,\,\, \vcenter{\hbox{\includegraphics[width=0.4\linewidth]{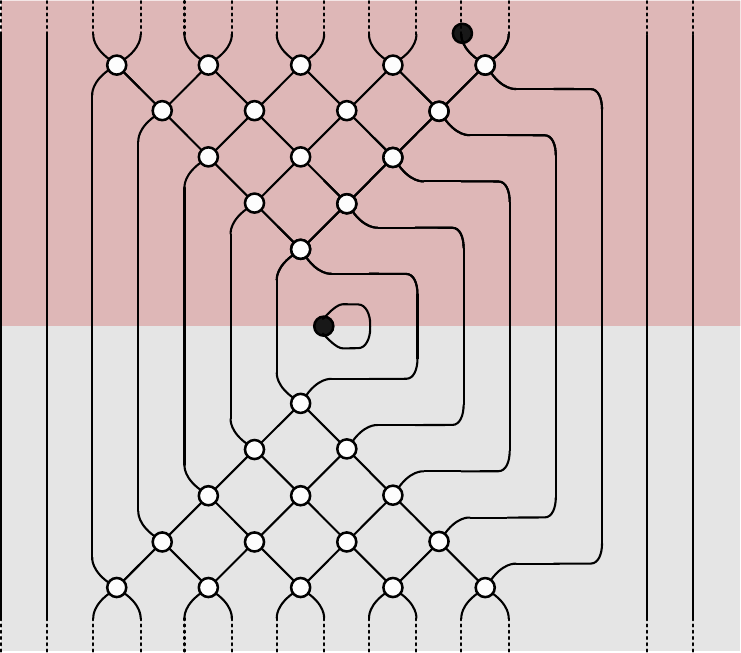}}}\,
\end{align}
As such, $\Tr(\rho)$ is again going to factorize out and the correlation function vanishes, $c_{\rho \sigma}(|x|<t-1) = 0$.

Exactly on the edge of the causal light cone the correlations can be calculated using the approach outlined in Refs.~\cite{bertini_exact_2019,claeys_maximum_2020}, using vertical unitarity to simplify the circuits as follows:
\begin{align}
\vcenter{\hbox{\includegraphics[width=0.4\linewidth]{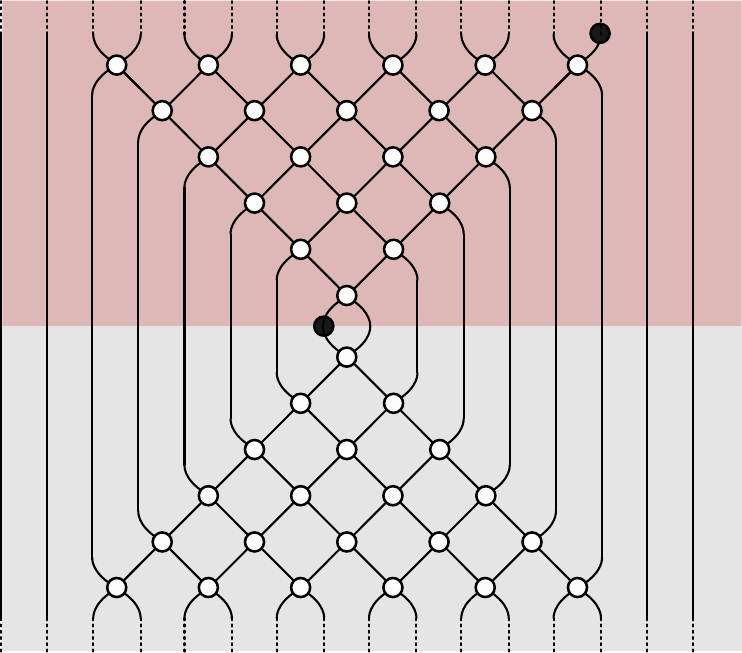}}} \,\,\,=\,\,\, \vcenter{\hbox{\includegraphics[width=0.4\linewidth]{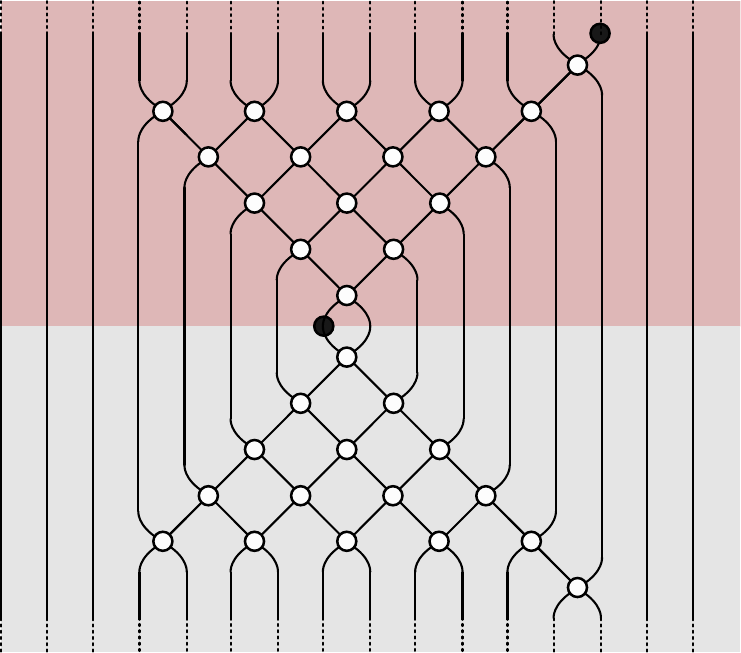}}}\,
\end{align}
Finally we obtain the following: 
\begin{align}
    c_{\rho\sigma}(x=t,t) \,\,\,=\,\,\, \,  \vcenter{\hbox{\includegraphics[width=0.4\linewidth]{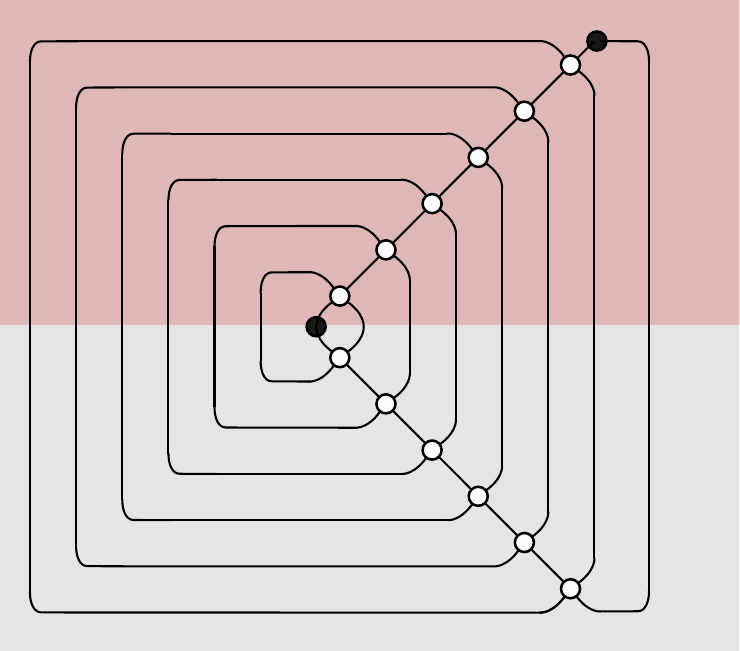}}}\,
\end{align}
The result for $|x|=t-1$ is analogous. Expressions of the above form can be efficiently calculated since the total number of biunitaries only grows linearly with the number of time steps (see again Refs.~\cite{bertini_exact_2019,claeys_maximum_2020}).
\end{proof}

We emphasize that all graphical manipulations are identical to the known derivations for dual-unitary brickwork circuits. The innovation here is to observe that the traditional proof extends to this more general case, thanks to the expressiveness of the shaded calculus. As a result, all such circuits will give rise to the light-cone dynamics of correlations functions that is characteristic of dual-unitary circuits.

%--------------------------------------------------------------------------------------------------------------------------------------------------------
\section{Entanglement dynamics}
\label{sec:entanglement}
\subsection{Solvable states}
Important properties of dual-unitary gates include maximal entanglement growth and exact thermalization to an infinite-temperature state after a finite number of time steps~\cite{gopalakrishnan_unitary_2019,bertini_operator_2020,piroli_exact_2020,reid_entanglement_2021,zhou_maximal_2022,foligno_growth_2022,claeys_exact_2022}. However, exact calculations of the entanglement growth and thermalization are restricted to special `solvable' initial states \cite{piroli_exact_2020}. In this Section, we first extend the notion of solvable states to general biunitary circuits, and subsequently use these to present exact results for entanglement dynamics and to prove that biunitary circuits exhibit exact thermalization after a finite number of time steps.

We first define the solvability of local tensors acting as building blocks for solvable states defined on the full lattice. 

\begin{defn}
A \textit{solvable tensor} is defined as a vertex $\mathcal{N}$, where all regions can again be shaded or not, 
\begin{align}
\mathcal{N} = \vcenter{\hbox{\includegraphics[width=0.15\linewidth]{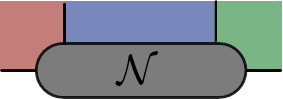}}}\,,
\end{align}
with a conjugate vertex $\mathcal{N}^{\dagger}$, such that the following conditions hold:
\begin{enumerate}
\item[(i)] These satisfy a notion of horizontal unitarity:
\begin{align}\label{eq:unitarity_N}
\vcenter{\hbox{\includegraphics[width=0.25\linewidth]{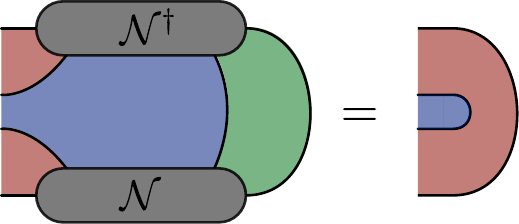}}} \qquad\qquad \vcenter{\hbox{\includegraphics[width=0.25\linewidth]{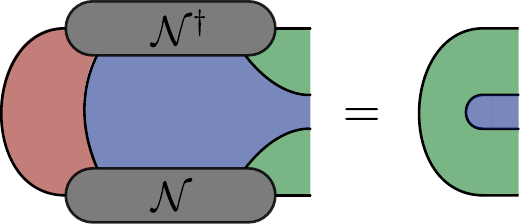}}}
\end{align}
\item[(ii)] The transfer matrix $E(\mathcal{N})$ constructed out of $\mathcal{N}$ and $\mathcal{N}^\dagger$ as
\begin{align}\label{eq:E_N}
E(\mathcal{N}) = \vcenter{\hbox{\includegraphics[width=0.15\linewidth]{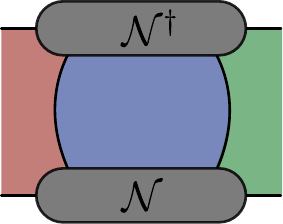}}}
\end{align}
can be rescaled to have a unique eigenvalue $\lambda$ with largest absolute value $\lambda=1$ and with algebraic multiplicity 1, with corresponding left and right eigenstates following from Eq.~\eqref{eq:unitarity_N} as follows:
\begin{align}\label{eq:eigenstates_N}
\vcenter{\hbox{\includegraphics[width=0.25\linewidth]{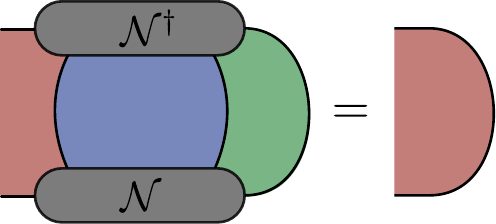}}} \qquad\qquad \vcenter{\hbox{\includegraphics[width=0.25\linewidth]{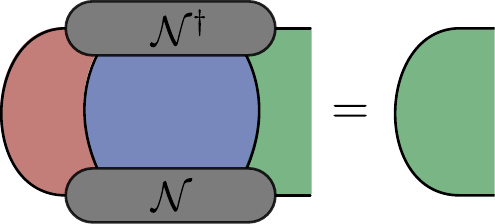}}}
\end{align}
\end{enumerate}
\end{defn}

This definition is a direct generalization of the definition of solvable matrix product states from Ref.~\cite{piroli_exact_2020}, now taking into account different possible shadings of the vertex. As will be illustrated below, the first condition generically implies the second one. 

A state for the full lattice can be constructed out of a set of solvable tensors as follows:
\begin{align*}
\ket{\Psi(\{\mathcal{N}\})} \,\,\,=\,\,\, \vcenter{\hbox{\includegraphics[width=0.55\linewidth]{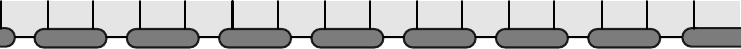}}}
\end{align*}
Here we have made the labels $\mathcal{N}$ implicit and again introduced a transparent gray background to indicate that every region in the above state can either be shaded or not, with the implicit assumption that every vertex corresponds to a solvable tensor. The resulting biunitary circuit dynamics of such a state can be represented in the following manner, with the restriction that the shading pattern of the initial state matches the shading pattern of the biunitary circuit:
\begin{align}
\ket{\Psi(t,\{\mathcal{N}\})} = \mathcal{U}_t\ket{\Psi(\{\mathcal{N}\})} \,\,\,=\,\,\, \vcenter{\hbox{\includegraphics[width=0.55\linewidth]{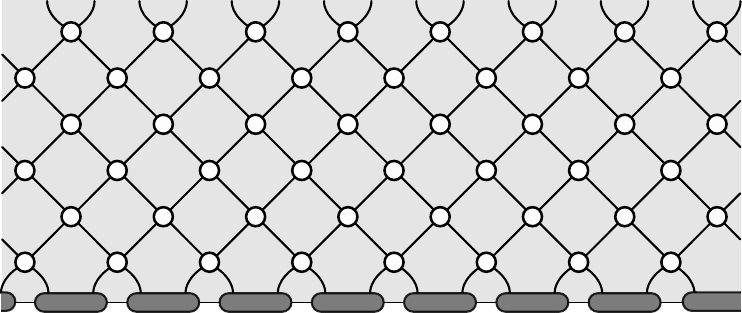}}}
\end{align}

\subsection{Entanglement entropy}
Quantum thermalization and entanglement is encoded in the reduced density matrix $\rho_A(t)$ for a subsystem $A$, which is defined by tracing out all degrees of freedom from the complement of $A$ in the density matrix. Graphically, the reduced density matrix can be represented as follows:
\begin{align}\label{eq:rdm}
    \rho_A(t)\,\,\,=\,\,\, \Tr_{\overline{A}} \left( \,\ket{\Psi(t,\{\mathcal{N}\})}\bra{\Psi(t,\{\mathcal{N}\})} \,\right) \,\,\,=\,\,\, \vcenter{\hbox{\includegraphics[width=0.4\linewidth]{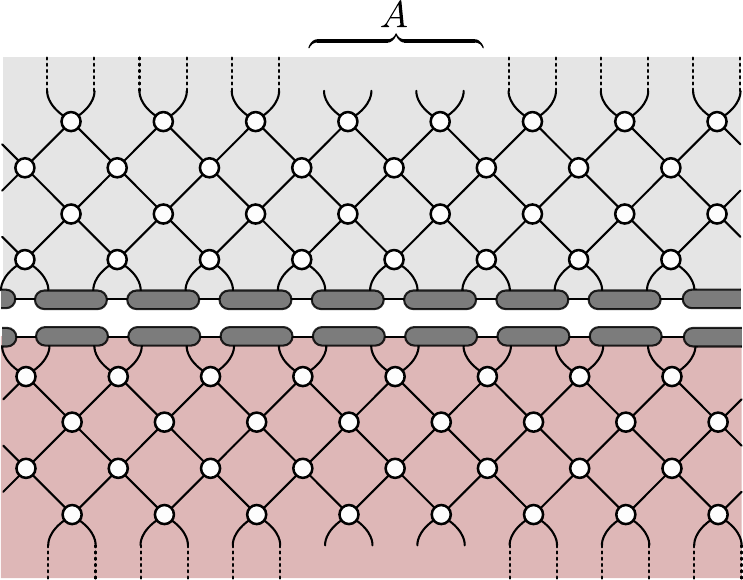}}}\,
\end{align}
Here we have chosen the subsystem $A$ to consist of 4 neighbouring wires, and have restricted the circuit to $t=4$ time steps. The dashed lines in the complement $A$ again indicate connections between the top region and the bottom region corresponding to the partial trace. Note that the Hilbert space of the subsystem $A$ will depend on the choice of shading.
\begin{theorem}
For a subsystem $A$ consisting of $\ell$ neighbouring wires, the reduced density matrix on $A$ for an initial solvable state evolved using a biunitary circuit for $t$ discrete time steps equals the identity operator on $A$ for $t \geq \ell/2$, i.e.
\begin{align*}
\rho_A(t \geq \ell/2) = \mathbbm{1}_A/q_A\,,
\end{align*}
where $q_A$ is the dimension of the Hilbert space of the subsystem $A$ in the circuit representation.

\end{theorem}
This result is known as \emph{thermalization}: after sufficiently long time evolution, the reduced density matrix for a subsystem $A$ attains a stationary value in which all information about the initial state is lost. Biunitary circuits hence exhibit exact thermalization of a subsystem $A$ after $t=\ell/2$ time steps. The entanglement entropy, defined as $S_A(t) = -\Tr[\rho_A(t) \ln(\rho_A(t))]$, reaches its maximal value of
\begin{align*}
S_A(t  \geq \ell/2) = \ln(q_A).
\end{align*}
Since the reduced density matrix equals the identity, this model has a flat entanglement spectrum~\cite{gopalakrishnan_unitary_2019}, indicating that all R\'enyi entropies $S_A^{(n)}(t)$, defined as
\begin{align}\label{eq:Renyi}
S_A^{(n)}(t) = \frac{1}{1-n} \ln \Tr\left[\rho_A(t)^n\right],
\end{align}
reach the same maximal value of $S_A^{(n)}(t  \geq \ell/2) = \ln(q_A), \forall n \in \mathbbm{N}$.

\begin{proof}
The proof can again be made purely graphical and follows the same steps as the derivation of Piroli \emph{et al.}~\cite{piroli_exact_2020}, with the shaded regions here giving greater generality. We here illustrate the proof for the setup of Eq.~\eqref{eq:rdm}, and extensions to different time steps and subsystems $A$ follow directly. Using vertical unitarity of the vertices, Eq.~\eqref{eq:rdm} can first be simplified to yield:
\begin{align}
\rho_A(t) \,\,\,=\,\,\, \vcenter{\hbox{\includegraphics[width=0.4\linewidth]{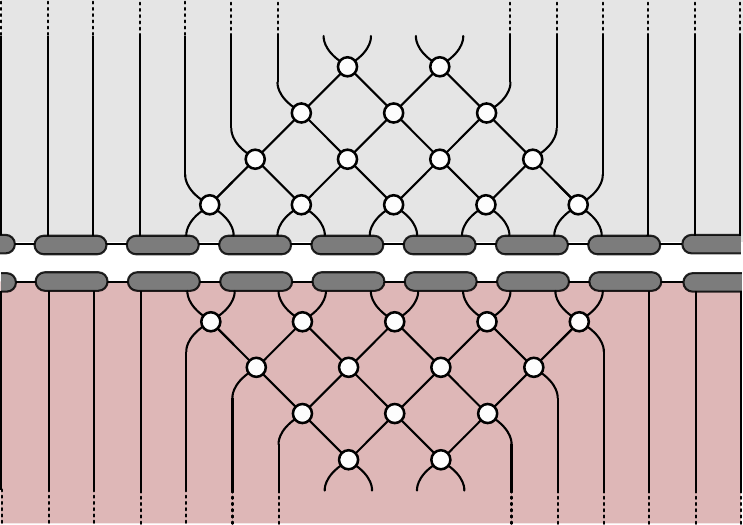}}}\,.
\end{align}
To the left and right, the effect of the environment now consists of the repeated application of the transfer matrix $E(\mathcal{N})$ from Eq.~\eqref{eq:E_N}. Since the environment is infinite, we can replace the transfer matrix by a projector on the leading eigenspace, which here is nondegenerate and with corresponding eigenstates defined in Eq.~\eqref{eq:eigenstates_N}. Introducing these eigenstates, the above expression simplifies as follows, where we have also moved all connecting wires along the top/bottom to the sides for clarity:
\begin{align}
\rho_A(t) \,\,\,=\,\,\, \vcenter{\hbox{\includegraphics[width=0.5\linewidth]{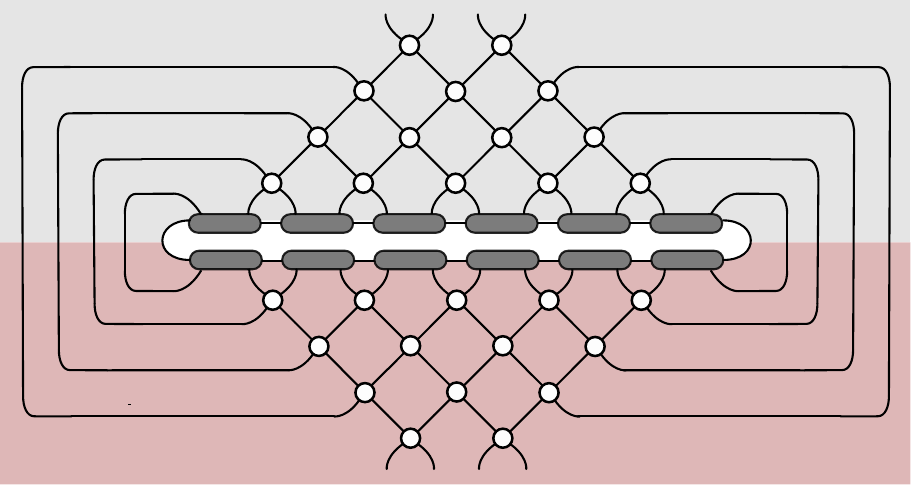}}}\,,
\end{align}
Using the solvability of the initial state, the horizontal contractions can be propagated along the initial state, using first the solvability condition and then the horizontal unitarity of the vertices in the circuit:
\begin{align}
\rho_A(t) = \vcenter{\hbox{\includegraphics[width=0.45\linewidth]{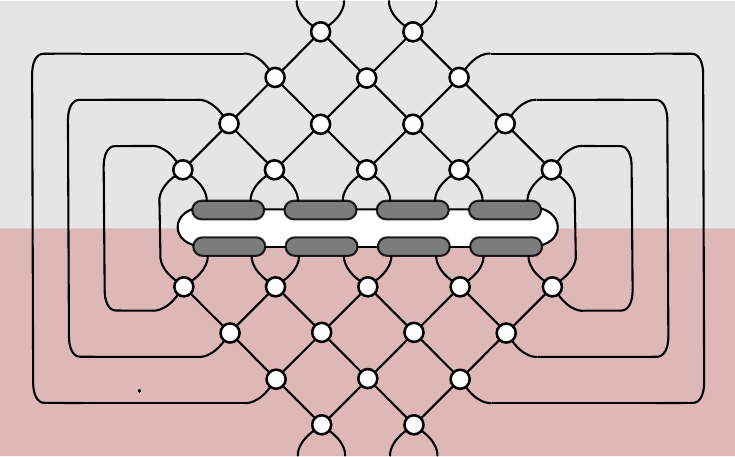}}}= \vcenter{\hbox{\includegraphics[width=0.45\linewidth]{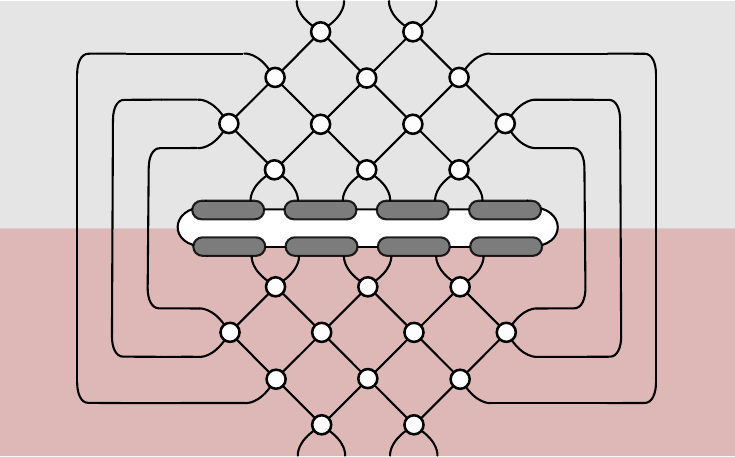}}}\,
\end{align}
This procedure can be repeated to eliminate all vertices originating from the initial state, leading to the following:
\begin{align}
\rho_A(t) \,\,\,=\,\,\, \vcenter{\hbox{\includegraphics[width=0.25\linewidth]{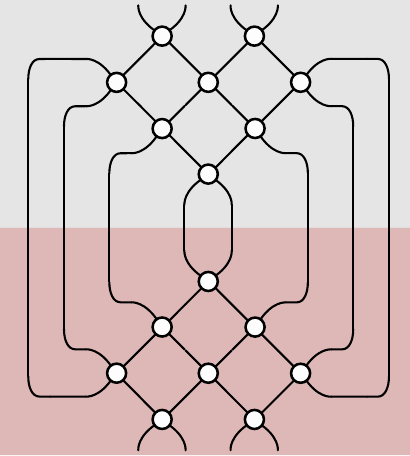}}}\,.
\end{align}
All remaining vertices can be removed using vertical unitarity, resulting in a final expression for the reduced density matrix as follows:
\begin{align}
\rho_A(t) \,\,\,=\,\,\, \vcenter{\hbox{\includegraphics[width=0.13\linewidth]{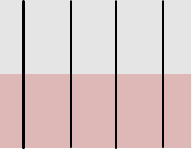}}}\, \,\,\,\propto\,\,\, \mathbbm{1}_A\,.
\end{align}
The correct prefactor can be recovered by noting that $\Tr[\rho_A(t)]=1$.
\end{proof}
The condition $t \geq \ell/2$ arises from the necessity to fully propagate the horizonal contractions along the initial state, which is only possible after sufficiently many time steps. To illustrate this condition, consider a subsystem of size $\ell=8$ and $t=2$ time steps, where a similar derivation can be performed:
\begin{align}\label{eq:rho_shortt}
\rho_A(t) \,\,\,=\,\,\, \vcenter{\hbox{\includegraphics[width=0.4\linewidth]{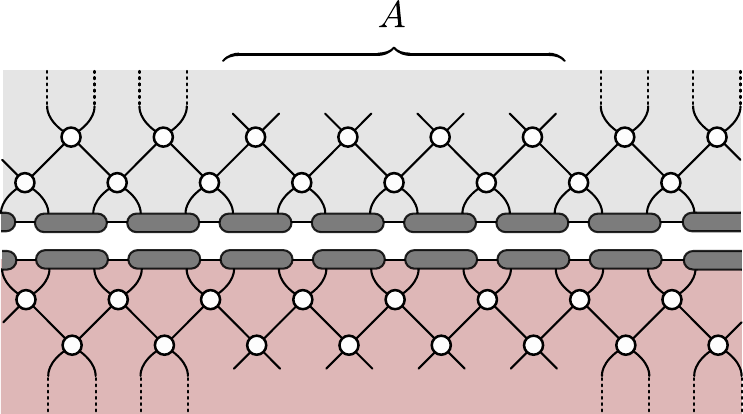}}}\,\,\,=\,\,\, \vcenter{\hbox{\includegraphics[width=0.24\linewidth]{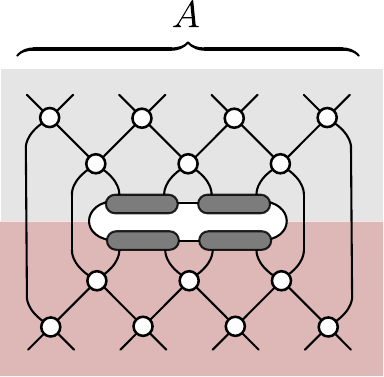}}}\,
\end{align}
The expression on the right hand side can no longer be simplified. For each time step we can perform a single contraction along the horizontal direction removing a two-site solvable tensor, resulting in the condition $t \geq \ell/2$. However, such an expression can still be used to calculate the entanglement entropy: Eq.~\eqref{eq:rho_shortt} corresponds to a unitary transformation of an operator acting nontrivially on $\ell-2t$ wires and acting as the identity on the remaining $2t$ wires. If the eigenspectrum of the nontrivial operator is known, the entanglement spectrum is known and the entanglement entropy can be calculated. Such calculations are typically possible in the scaling limit where $\ell, t \to \infty$ while keeping the ratio $\ell/t$ fixed (see Ref.~\cite{piroli_exact_2020} and below).

\subsection{Examples}
In this subsection we consider the two limiting cases of brickwork and clockwork circuits as well as the intermediate dynamical boundary construction and explicitly construct solvable states as matrix product states (MPSs). The explicit representation as matrix product states guarantees that these states are area-law entangled \cite{orus_practical_2014} (the entanglement $S_A$ does not scale with the size of $A$ but rather with the size of the boundary $A$, i.e. is constant in one dimension), to be contrasted with the volume-law entanglement of the steady-state reduced density matrix (the entanglement scales with the size of $A$).

\paragraph{Matrix product states.} For completeness, we briefly recall the definition of matrix product states (see e.g. Ref.~\cite{orus_practical_2014}). A two-site shift invariant matrix product state for a system of $L$ sites with local Hilbert space dimension $q$ can be written as
\begin{align}
\ket{\Psi(\{A\},\{B\})} = \sum_{\{i_j\}}^q  \Tr\left[A^{(i_1)} B^{(i_2)} A^{(i_3)} B^{(i_4)} \dots  \right]\ket{i_1, i_2 \dots i_L},
\end{align}
parametrized in terms of a sets of matrices $\{A^{(i)}|i=1\dots q\}$, $\{B^{(i)}|i=1\dots q\}$ with dimensions $\chi \times \chi$. The dimension $\chi$ is known as the bond dimension and determines an upper bound on the entanglement between lattice subsystems. Crucially, if $\chi$ is a constant the entanglement entropy of any subsystem does not scale with subsystem size.

The amplitudes in the wave function can be graphically expressed by introducing the tensors
\begin{calign}\label{eq:MPS}
A^{(i)}_{ab} = \, \vcenter{\hbox{\includegraphics[width=0.1\linewidth]{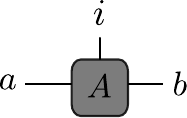}}}
&
B^{(i)}_{ab} =\, \vcenter{\hbox{\includegraphics[width=0.1\linewidth]{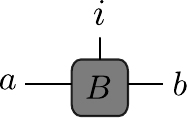}}}\,
\end{calign}
which results in:
\begin{align}
\ket{\Psi(\{A\},\{B\})} = \sum_{\{i_j\}}^q\, \vcenter{\hbox{\includegraphics[width=0.4\linewidth]{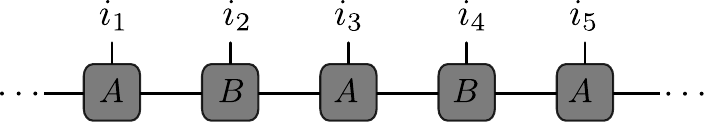}}}\, \ket{i_1, i_2 \dots i_L}
\end{align}

\paragraph{Brickwork circuits.} For purely dual-unitary circuits the definition of solvable tensors by construction reproduces the notion of exactly solvable matrix product states \cite{piroli_exact_2020}. 
Since there are no shaded areas, these tensors have four indices and we can write:
\begin{align}
\mathcal{N}^{(b,c)}_{ad} = \vcenter{\hbox{\includegraphics[width=0.18\linewidth]{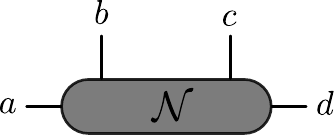}}}.
\end{align}
The indices $b,c$ correspond to the physical Hilbert space with e.g. dimension $q$, whereas the indices $a,d$ correspond to an auxiliary Hilbert space that we take to have bond dimension $\chi$.
The corresponding initial state is immediately expressed as a matrix product state. For convenience we take all $\mathcal{N}$ to be identical, resulting in a translationally invariant state, although this is not a necessary restriction:
\begin{align}
\ket{\Psi(\mathcal{N})} &= \sum_{\{i_j\}}^q \vcenter{\hbox{\includegraphics[width=0.43\linewidth]{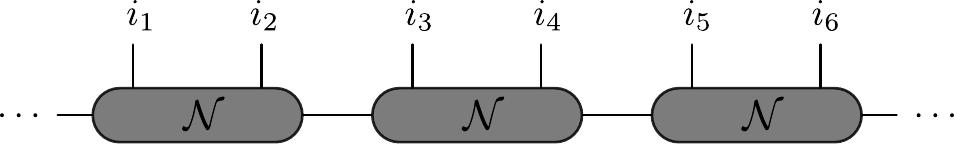}}} \ket{\dots i_1\, i_2\, i_3\, i_4\, i_5\, i_6 \dots} \\
&=\sum_{\{i_j\}}^q \Tr\left[\dots \mathcal{N}^{(i_1,i_2)}\mathcal{N}^{(i_3,i_4)}\mathcal{N}^{(i_5,i_6)} \dots \right]\ket{\dots i_1\, i_2\, i_3\, i_4\, i_5\, i_6 \dots}
\end{align}
This expression corresponds to a matrix product state \eqref{eq:MPS} with $\mathcal{N}^{(i,j)} = A^{(i)}B^{(j)}$.
The first solvability condition can be written as:
\begin{align}
\vcenter{\hbox{\includegraphics[width=0.3\linewidth]{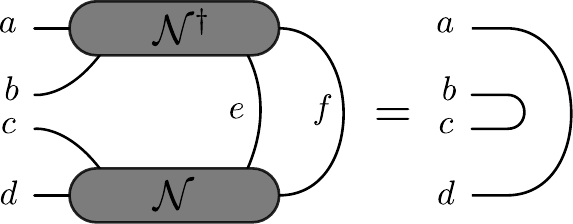}}} \qquad \Rightarrow \qquad \sum_{e=1}^q\sum_{f=1}^{\chi} \left[\mathcal{N}^{(b,e)}_{af}\right]^*\mathcal{N}^{(c,e)}_{df} = \delta_{ad}\delta_{bc}.
\end{align}
Such states have been extensively analyzed in Ref.~\cite{piroli_exact_2020}, and we here repeat the main results. The condition \eqref{eq:unitarity_N} corresponds to unitarity of the matrix $\mathcal{W}$ with matrix elements $\mathcal{W}_{ab,cd} = \mathcal{N}^{(b,c)}_{ad}$. As such, solvable state are uniquely parametrized by unitary matrices acting on $\mathbbm{C}^q \otimes \mathbbm{C}^\chi$. The second condition typically needs to be checked on a case by case basis and corresponds to injectivity of the corresponding MPS. Explicit parametrization for different bond dimensions $\chi=1$ and $\chi=2$ were presented in Ref.~\cite{piroli_exact_2020}, and for $\chi=1$ injectivity is guaranteed whereas for $\chi=2$ the non-injective states form a set of measure zero. As such, the first condition \eqref{eq:unitarity_N} is the more stringent one.

The time-evolved state can be represented in the shaded calculus and in tensor network notation respectively as:
\begin{align}
\vcenter{\hbox{\includegraphics[width=0.7\linewidth]{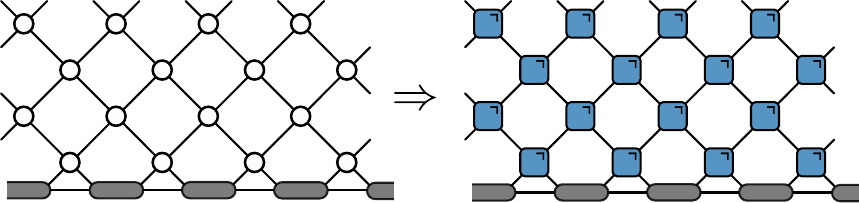}}}
\end{align}
As shown in Ref.~\cite{piroli_exact_2020}, the entanglement entropy for a subsystem of size $\ell$ and local Hilbert space dimension $q$ can be calculated in the scaling limit where $\ell, t \to \infty$ for a finite ratio $\ell/t = \zeta$, using expressions of the form \eqref{eq:rho_shortt},where it was shown to take the universal form
\begin{align}\label{eq:entgrowth_bw}
\lim_{\substack{\ell,t \to \infty \\ \ell/t = \zeta}} S_A(t)/\ell = \textrm{min}(2,\zeta^{-1})\log(q)\,.
\end{align}
The entanglement entropy grows with the maximal possible slope of $2\log(q)$, characteristic of dual-unitary circuits \cite{zhou_maximal_2022}.

\paragraph{Clockwork circuits.} Our construction directly allows for the construction of solvable initial states for clockwork circuits, which we introduce here. Since all regions are now shaded, the tensors $\mathcal{N}$ are parametrized by three indices, and we write:
\begin{align}\label{eq:N_cw}
\mathcal{N}^{(b)}_{ac} = \vcenter{\hbox{\includegraphics[width=0.18\linewidth]{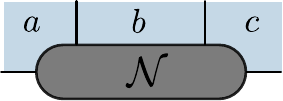}}}\,.
\end{align}
We take all indices to correspond to dimension $q$, since these will correspond to physical indices in the tensor network representation.
The first solvability conditions are as follows:
\begin{align}
\vcenter{\hbox{\includegraphics[width=0.3\linewidth]{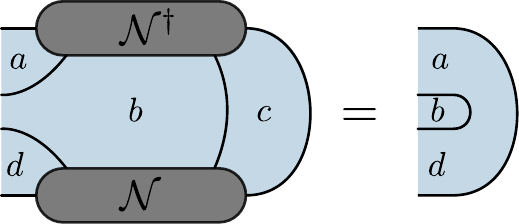}}} \qquad \Rightarrow \qquad \sum_{c=1}^q \left[\mathcal{N}^{(b)}_{ac}\right]^*\mathcal{N}^{(b)}_{dc} = \delta_{ad}, \qquad \forall b=1\dots q
\\[5pt]
\vcenter{\hbox{\includegraphics[width=0.3\linewidth]{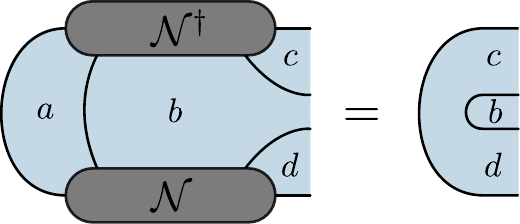}}} \qquad \Rightarrow \qquad \sum_{a=1}^q \left[\mathcal{N}^{(b)}_{ac}\right]^*\mathcal{N}^{(b)}_{ad} = \delta_{cd}, \qquad \forall b=1\dots q
\end{align}
These conditions can be directly seen to be equivalent to the unitarity of the matrices $\mathcal{N}^{(b)}, \forall b$. As such, solvable states for clockwork circuits are parametrized by a set of unitary matrices $\{\mathcal{N}^{(b)} \in U(q) |b=1\dots q\}$.

For any such set of unitary matrices the additional condition on the transfer matrix \eqref{eq:E_N} is generically satisfied. This can be understood from an explicit construction: the matrix elements of the $q\times q$ transfer matrix $E(\mathcal{N})$ are defined from Eq.~\eqref{eq:E_N} as
\begin{align}
    E(\mathcal{N})_{ac} = \frac{1}{q}\sum_{b=1}^q \left[\mathcal{N}^{(b)}_{ac}\right]^*\mathcal{N}^{(b)}_{ac} = \frac{1}{q}\sum_{b=1}^q \left|\mathcal{N}^{(b)}_{ac}\right|^2
\end{align}
where we have rescaled $E(\mathcal{N})$ by a factor $q$. Due to the unitarity of the individual matrices $\mathcal{N}^{(b)}$ this is a doubly stochastic matrix. Doubly stochastic matrices are guaranteed to have leading eigenvalue 1, with corresponding (left and right) eigenvectors $v$ satisfying $v_c = 1, \forall c$, consistent with Eq.~\eqref{eq:eigenstates_N}. Furthermore, if all matrix elements are positive the Perron–Frobenius theorem guarantees the nondegeneracy of this leading eigenvalue. Nonapplicability of the Perron-Frobenius theorem would require zero matrix elements in the transfer matrix, i.e. $E(\mathcal{N})_{ac} =0$, in turn requiring the corresponding matrix element to vanish in all unitary matrices and $\mathcal{N}^{(b)}_{ac}=0, \forall b$. As such, generic sets of unitary matrices result in solvable initial states.

Explicitly writing out the corresponding state results in the following, where for convenience we have again taken all $\mathcal{N}$ to be identical:
\begin{align}
\ket{\Psi(\mathcal{N})} &= \sum_{i_1, i_2,\dots}^q \vcenter{\hbox{\includegraphics[width=0.45\linewidth]{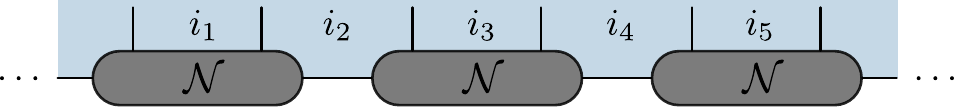}}}\, \ket{\dots i_1\, i_2\, i_3\, i_4\, i_5 \dots} \\
&=\sum_{\dots, i_1, i_2\dots}^q \left[\dots \mathcal{N}^{(i_1)}_{\dots i_2 }\mathcal{N}^{(i_3)}_{i_2,i_4}\mathcal{N}^{(i_5)}_{i_4\dots} \dots \right]\ket{\dots i_1\, i_2\, i_3\, i_4\, i_5\, i_6 \dots}\,,
\end{align}
This state can be rewritten as an exact MPS by introducing auxiliary tensors
\begin{calign}
\vcenter{\hbox{\includegraphics[width=0.15\linewidth]{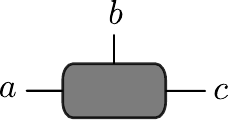}}}  = \mathcal{N}^{(b)}_{ac}
 &
 \vcenter{\hbox{\includegraphics[width=0.11\linewidth]{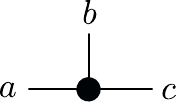}}}  = \delta_{a,b,c}
\end{calign}
where the second tensor serves to take care of the threefold appearances of the indices $i_2,i_4,\dots$ in the argument of the summation. These correspond to the $A$ and $B$ tensors in Eq.~\eqref{eq:MPS}, resulting in the equivalent MPS representation:
\begin{align}
\ket{\Psi(\mathcal{N})} &= \sum_{i_1, i_2,\dots}^q  \vcenter{\hbox{\includegraphics[width=0.45\linewidth]{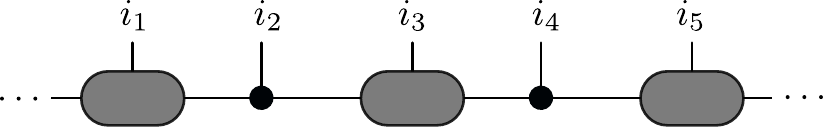}}}\, \ket{\dots i_1\, i_2\, i_3\, i_4\, i_5 \dots}
\end{align}
The time-evolved state can be represented in the shaded calculus and in tensor network notation respectively as follows:
\begin{align}
\vcenter{\hbox{\includegraphics[width=0.7\linewidth]{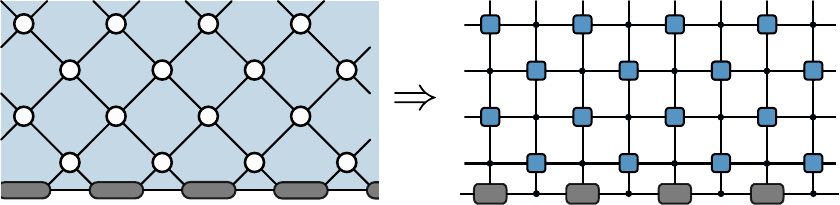}}}
\end{align}
The entanglement entropy for a subsystem of size $\ell$ and local Hilbert space dimension $q$ can again be calculated in the scaling limit where $\ell, t \to \infty$ for a finite ratio $\ell/t = \zeta$, using expressions of the form \eqref{eq:rho_shortt}:
\begin{align}\label{eq:entgrowth_bw2}
\lim_{\substack{\ell,t \to \infty \\ \ell/t = \zeta}} S_A(t)/\ell = \textrm{min}(2,\zeta^{-1})\log(q)
\end{align}
Since the derivation for the entanglement entropy is similar to the known derivation for dual-unitary circuits~\cite{piroli_exact_2020}, we here only sketch the outline. The R\'enyi entropies \eqref{eq:Renyi} can be calculated from $\Tr[\rho_A(t)^n]$, the graphical expression of which contains contractions of the transfer matrix~\eqref{eq:E_N}. In the scaling limit the transfer matrix can be replaced by the projector on its leading eigenstates and all contractions explicitly evaluated, resulting in an expression for the trace as $q^{(2t+2)(1-n)}$, in turn giving the result for the presented entanglement entropy. Clockwork circuits reproduce the maximal entanglement growth of dual-unitary brickwork circuits.

\paragraph{Hybrid circuits.} More general biunitary circuits typically require solvable states combining shaded and unshaded regions. The corresponding tensors in the initial state are largely similar to the ones in the clockwork circuits. With two shaded regions, the tensors $\mathcal{N}$ are again parametrized by three indices, and we write
\begin{align}
\mathcal{N}^{(b)}_{ac} = \vcenter{\hbox{\includegraphics[width=0.18\linewidth]{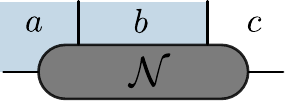}}}\,,
\end{align}
and the conditions for solvability are formally equivalent to those for Eq.~\eqref{eq:N_cw}:
\begin{align}
\vcenter{\hbox{\includegraphics[width=0.3\linewidth]{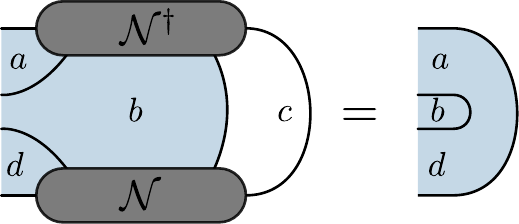}}} \qquad \Rightarrow \qquad \sum_{c=1}^q \left[\mathcal{N}^{(b)}_{ac}\right]^*\mathcal{N}^{(b)}_{dc} = \delta_{ad}, \qquad \forall b=1\dots q,
\\[5pt]
\vcenter{\hbox{\includegraphics[width=0.3\linewidth]{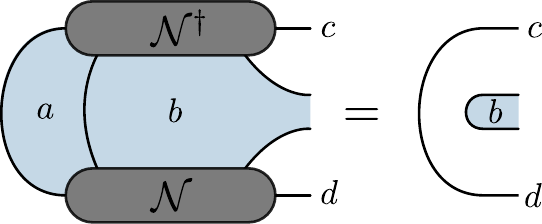}}} \qquad \Rightarrow \qquad \sum_{a=1}^q \left[\mathcal{N}^{(b)}_{ac}\right]^*\mathcal{N}^{(b)}_{ad} = \delta_{cd}, \qquad \forall b=1\dots q.
\end{align}
As such, the corresponding states can again be parametrized by a set of unitary matrices and be represented as a matrix product state. Considering e.g. the case of a diagonal moving boundary between the two regions, we find that:
\begin{align}
\vcenter{\hbox{\includegraphics[width=0.7\linewidth]{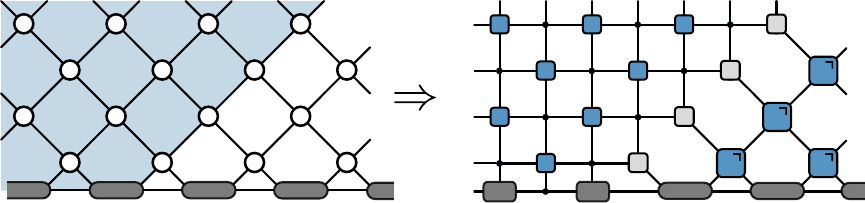}}}
\end{align}

\section{Restrictions on the Hilbert space dimension}
\label{sec:dimensions}
As already mentioned, the inclusion of unitary error bases introduces relations between the different Hilbert space dimensions. These restrictions can be illustrated by considering a dual-unitary brickwork circuit where a clockwork wedge has been inserted:
\begin{align}
\vcenter{\hbox{\includegraphics[width=0.5\linewidth]{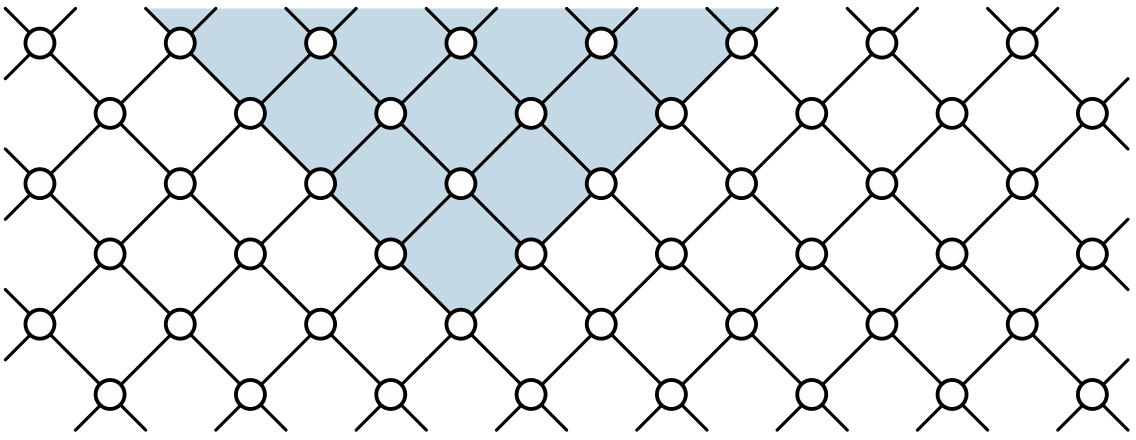}}},
\end{align}
Such diagrams now consist of both dual-unitary gates, quantum crosses, quantum Latin squares, and a single unitary error basis. When represented in terms of unitary gates, this circuit is represented as follows:
\begin{align}
\vcenter{\hbox{\includegraphics[width=0.5\linewidth]{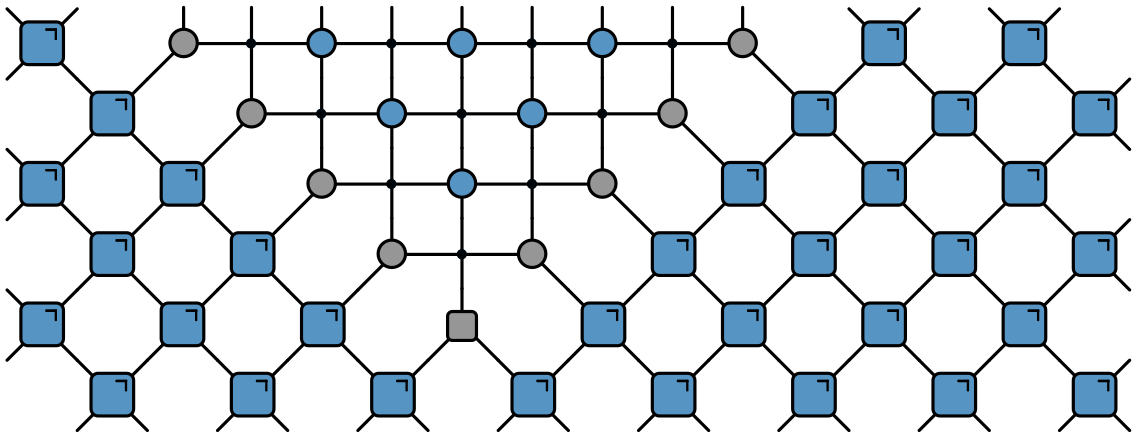}}},
\end{align}
The introduction of unitary error bases introduces restrictions on the local Hilbert spaces, since in this example the UEB has two wires of dimension $q$ coming in and a single wire of dimension $q^2$ coming out. However, circuits where wires carry different Hilbert spaces can still be considered, as also argued in Ref.~\cite{Borsi2022} for dual-unitary brickwork circuits.

The general restrictions on the dimensions of Hilbert spaces imposed by each biunitary can be formulated as follows (see also Section~\ref{sec:quantumstructures}): for dual-unitary gates, diametrically opposite wires must have the same dimension; for quantum crosses, opposite shaded regions must have the same dimension; for a quantum Latin square, the wire must have the same dimension as both shaded regions; for a Hadamard matrix, both shaded regions have the same dimension. For a UEB, the wires must have the same dimension as each other, and the shaded region must have the dimension of the product space.

In the above example, if we label the incoming wires in the UEB with Hilbert space dimensions $q$, then the outcoming shaded region must have dimension $q^2$. Since the UEB borders on a QLS, where both the two shaded regions and the single wire have the same dimension, this determines the Hilbert space dimensions of all wires bordering on a QLS to have dimension $q^2$. The quantum cross bordering the UEB similarly enforces all shaded regions inside the wedge to have dimension $q^2$, as illustrated below.
\begin{align}
\vcenter{\hbox{\includegraphics[width=0.35\linewidth]{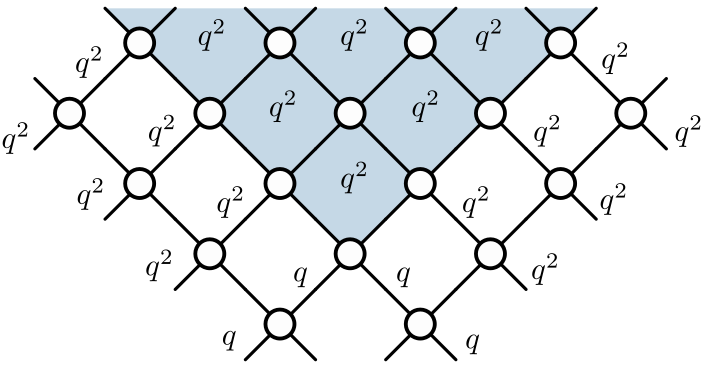}}},
\end{align}
%%
\begin{comment}
Despite these different local Hilbert spaces, it is possible to construct initial solvable states if the Hilbert spaces of neighbouring sites are pairwise identical. We can consider a more complicated example for which it is still possible to construct such initial solvable states. Consider a dual-unitary diagram where a single rectangle of clockwork is embedded in a dual-unitary brickwork circuit,
%%
\begin{align}
\vcenter{\hbox{\includegraphics[width=0.5\linewidth]{state_1}}}\,.
\end{align}
%%
Due to the inclusion of a UEB not all wires can carry the same Hilbert space. However, it it still possible to pair neighbouring sites, as illustrated below. In order to avoid a cluttering of labels we draw wires carrying a $q$-dimensional Hilbert space using a normal line and wires carrying a $q^2$-dimensional Hilbert space using a thick line. 
%%
\begin{align}
\vcenter{\hbox{\includegraphics[width=.5\linewidth]{state_2}}}\,.
\end{align}
%%
All Hilbert space dimensions are consistent with biunitarity. It is a straightforward calculation to check that for $t$ rows the reduced density matrix for $N \leq 2t$ sites will again equal the identity. However, in this example the Hilbert space dimension of the subsystem can change in time.
\end{comment}
%%

\paragraph{Trivialisation.}
Note that it is possible to have restrictions where the only consistent solution is to have a large part of the local Hilbert spaces to be one-dimensional, in which case part of the circuit trivializes. We illustrate one such a circuit below:
\begin{align}
\vcenter{\hbox{\includegraphics[width=0.4\linewidth]{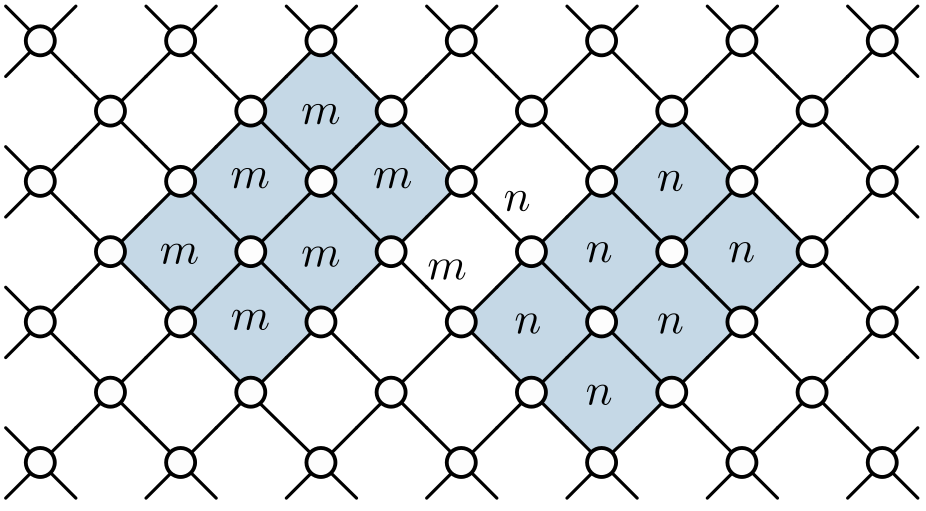}}}
\end{align}
In this example, the quantum crosses and Latin squares force the shaded regions to have the same dimension within each rectangle. The UEBs in the middle of the picture require $n^2=m$ and $m^2=n$, for which the only integer solution is $m=n=1$. 

%--------------------------------------------------------------------------------------------------------------------------------------------------------
\section{Discussion and outlook}
\label{sec:conclude}

Through the use of a 2\-categorical framework we have introduced biunitary circuits, unifying the notions of dual-unitary brickwork and interactions round-a-face (clockwork) circuits. While in this work we focused on the dynamics of correlation functions and entanglement, we expect that general results for dual-unitary circuits can be extended to biunitary circuits. Dual-unitary circuits have gained special attention in the context of quantum chaos because of exact calculations of e.g. the spectral form factor \cite{bertini_exact_2018,bertini_random_2021} and the spectral function \cite{fritzsch_eigenstate_2021} indicating chaotic behaviour. We expect that similar calculations are possible for biunitary circuits using the presented graphical calculus, which would establish the connection between biunitary circuits and random matrix theory. Furthermore, while generic parametrizations of dual-unitary gates will generally lead to chaotic quantum dynamics, it is possible to fine-tune the underlying gates such that the full circuit is no longer chaotic and systematic parametrizations of dual-unitary gates remain an active topic of research \cite{bertini_exact_2019,bertini_operator_ii_2020,claeys_ergodic_2021,Borsi2022,aravinda_dual-unitary_2021,rather_construction_2022,gombor_superintegrable_2022,singh_ergodic_2022,brahmachari_dual_2022}. It is still an open question how parametrizations of these new biunitary connections can influence the level of ergodicity and the corresponding diagnostics of quantum chaos. Additionally, exact results on operator spreading in (perturbed) dual-unitary brickwork circuits should have direct analogs in the presented biunitary circuits \cite{bertini_scrambling_2020,bertini_operator_2020,claeys_maximum_2020,reid_entanglement_2021,Rampp2022}.

Measurement schemes inspired by biunitarity have already allowed for exact studies of the interplay of dual-unitary dynamics with projective measurements \cite{claeys_exact_2022} and established further connections with random matrix theory through \emph{deep thermalization} \cite{ho_exact_2022, claeys_emergent_2022, ippoliti_dynamical_2022,claeys_universality_2023}. Exact calculations of the latter depended on the definition of solvable measurement schemes as an extension of solvable initial states preserving spatial unitarity, and it is an open question how to extend these solvable measurement schemes to general biunitary circuits. This question is especially relevant given the possible connection with measurement-based quantum computation enabled by dual-unitarity \cite{Stephen2022}.

There are several natural avenues along which to explore generalizations of the circuits described here, including the effects of lattice symmetries \cite{jonay_triunitary_2021,sommers2022crystalline,mestyan2022multi}, discrete conformal symmetry~\cite{Masanes2023}, random geometries \cite{kasim_dual_2022}, open systems \cite{kos_circuits_2022}, and higher dimensions~\cite{suzuki_computational_2022,Milbradt2023}.
On the level of biunitarity, it is possible to further extend the biunitary building blocks and include controlled families of biunitary connections~\cite{reutter_biunitary_2019}. Such biunitaries can similarly be included in the presented construction and would correspond to the propagation of classical information. We leave the study of such circuits to future work.

%--------------------------------------------------------------------------------------------------------------------------------------------------------
\section*{Acknowledgments}
A.L. gratefully acknowledges support from EPSRC Grant No. EP/P034616/1. J.V. gratefully acknowledges funding from the Royal Society. P.W.C. gratefully acknowledges M.A. Rampp for useful comments on the manuscript.

\appendix
%--------------------------------------------------------------------------------------------------------------------------------------------------------
\section{Explicit parametrizations of composite biunitaries}
\label{app:param_KIM}

In this Appendix we provide some explicit parametrizations for the composite biunitary connections discussed in Sec.~\ref{subsec:periodic}. The dual-unitary gates from Eq.~\eqref{eq:KIM_DU} are written as
\begin{align}
U_{ab,cd} =\,\, \vcenter{\hbox{\includegraphics[width=0.4\linewidth]{fig_KIM_DU}}}\,.
\end{align}
While the four complex Hadamard matrices can be chosen to be different, we will focus on the case where these are identical. For $q=2$ the simplest complex Hadamard matrices is given by
\begin{align}
H = \frac{1}{\sqrt{2}}\begin{pmatrix}
1 & 1 \\
 1 & -1
\end{pmatrix}\,.
\end{align}
The resulting dual-unitary gate has matrix elements $U_{ab,cd} = q \times H_{ab}H_{bd}H_{dc}H_{ca} $, which evaluates to
\begin{align}
U = \frac{1}{2}
\begin{pmatrix}
1 & 1 & 1 & -1 \\
1 & -1 & 1 & 1 \\
1 & 1 & -1 & 1 \\
-1 & 1 & 1 & 1
\end{pmatrix}\,.
\end{align}
These dual-unitary gates have the special properties that they are both Clifford gates and result in `matchgate' circuits that can be mapped to free fermions \cite{Valiant2001,Terhal2002}. The resulting circuit corresponds to the self-dual kicked Ising model at the integrable point \cite{akila_particle-time_2016,bertini_entanglement_2019,gopalakrishnan_unitary_2019,ho_exact_2022,Stephen2022}. Note that this gate is by construction not just dual-unitary, but self-dual: $\tilde{U}=U$.

The quantum cross from Eq.~\eqref{eq:KIM_cross} is constructed as
\begin{align}
(U_{a,c})_{b,d} = \,\,\, \vcenter{\hbox{\includegraphics[width=0.45\linewidth]{fig_KIM_cross}}}\,.
\end{align}
Expressed in matrix elements $(U_{a,c})_{b,d} = q \times \sum_{e} H_{ae}H_{be}H_{ce}H_{de}$, the different matrices evaluate to
\begin{align}
U_{0,0} = U_{1,1} =  \begin{pmatrix}
1 & 0 \\ 
0 & 1
\end{pmatrix}, \qquad U_{1,0} = U_{0,1} = \begin{pmatrix}
0 & 1 \\ 1 & 0
\end{pmatrix}\,.
\end{align}
More involved biunitary connections follow from dressing the complex Hadamard matrix with a phase, e.g. taking
\begin{align}
H = \frac{1}{\sqrt{2}}\begin{pmatrix}
1 & e^{i \phi} \\
1 & -e^{i \phi}
\end{pmatrix},
\end{align}
with the resulting dual-unitary gate given by 
\begin{align}
U = \frac{1}{2}
\begin{pmatrix}
1 & e^{2i\phi} & e^{i\phi} & -e^{i\phi} \\
e^{i\phi} & -e^{i\phi} & e^{2i\phi} & e^{4i\phi} \\
1 & e^{2i\phi}& -e^{i\phi} & e^{i\phi} \\
-e^{i\phi} & e^{i\phi} & e^{2i\phi} & e^{4i\phi}
\end{pmatrix}\,,
\end{align}
and the quantum cross by
\begin{align}
U_{0,0} = U_{1,1} =  e^{2i\phi} 
\begin{pmatrix}
\cos2\phi & -i \sin2\phi \\ 
-i\sin2\phi & \cos2\phi
\end{pmatrix}, \qquad U_{1,0} = U_{0,1} = e^{2i\phi} 
\begin{pmatrix}
-i \sin2\phi  & \cos2\phi \\ \cos2\phi & -i\sin2\phi
\end{pmatrix}\,.
\end{align}
For a generic phase the resulting circuit now corresponds to the self-dual kicked Ising model away from the integrable point.

%--------------------------------------------------------------------------------------------------------------------------------------------------------

%\bibliographystyle{plainurl}
\bibliographystyle{MyBibTexStyle}
\bibliography{Library}

\end{document}